\documentclass[preprint,12pt]{elsarticle}

\usepackage{amssymb}

\usepackage{lineno}

\usepackage{adjustbox}
\usepackage{booktabs}
\usepackage{multirow}
\usepackage{makecell}
\usepackage{bigdelim}
\usepackage{wrapfig}
\usepackage{colortbl}
\usepackage{etex}
\usepackage[ruled]{algorithm2e}
\usepackage{amsmath}
\usepackage{algorithmic}
\usepackage{amsfonts}
\usepackage{amssymb}
\usepackage{amsthm}
\usepackage{latexsym}
\usepackage{graphicx}
\usepackage{color}
\usepackage[normalem]{ulem}
\usepackage{balance}
\usepackage{textcomp}
\usepackage{nicefrac}
\usepackage{wasysym}
\usepackage[inline]{enumitem}
\usepackage{gensymb}
\usepackage{url}
\usepackage{longtable}
\usepackage{frcursive}
\usepackage{pifont,epsfig,rotating}
\usepackage{etoolbox}
\usepackage{etaremune}
\usepackage{mathtools}

\definecolor{dgreyblue}{rgb}{0.26,0.3,0.46}             

\newcommand{\cA}{\mathcal{A}}
\newcommand{\cB}{\mathcal{B}}

\newcommand{\cE}{\mathcal{E}}

\newcommand{\cI}{\mathcal{I}}

\newcommand{\cM}{\mathcal{M}}
\newcommand{\mopt}{{m}_{\mathrm{opt}}}
\newcommand{\Mopt}{{M}_{\mathrm{opt}}}

\newcommand{\cT}{\mathcal{T}}

\newcommand{\R}{{\mathbb R}}  
\newcommand{\PP}{{\mathbb P}}  

\renewcommand{\text}[1]{\hbox{\rm \ #1\ \/}}

\newcommand{\be}[1]{\begin{equation}\label{#1}}
\newcommand{\ee}{\end{equation}}
\newcommand{\beqn}{\begin{eqnarray*}}
\newcommand{\eeqn}{\end{eqnarray*}}
\newcommand{\beq}{\begin{eqnarray}}
\newcommand{\eeq}{\end{eqnarray}}
\newcommand{\ben}{\begin{enumerate}}
\newcommand{\een}{\end{enumerate}}
\newcommand{\bi}{\begin{itemize}}
\newcommand{\ei}{\end{itemize}}
\newcommand{\eps}{\varepsilon}

\newcommand{\IE}{{\em i.e.}\xspace}
\newcommand{\tx}{^{\rm th}}
\newcommand{\Ave}[1]{\mathbb{E}\left[ #1\right]}

\newtheorem{fact}{Fact}

\newtheorem{theorem}{Theorem}
\newtheorem{remark}{Remark}
\newtheorem{lemma}[theorem]{Lemma}
\newtheorem{corollary}[theorem]{Corollary}
\newtheorem{definition}[theorem]{Definition}

\newenvironment{proof-sketch}{{\noindent\bf Sketch of Proof.\ }}{\hfill{\Pisymbol{pzd}{113}}\vspace{0.1in}}

\newcommand{\LP}{\mathsf{LP}}

\newcommand{\cS}{\mathcal{S}}
\newcommand{\M}{\mathcal{M}}
\newcommand{\nbr}{\mathsf{Nbr}}

\newcommand{\EA}{{\em et al.}\xspace}
\newcommand{\TB}{\vspace{-0.1ex}}\newcommand{\TiE}{\setlength{\itemsep}{-1ex}}

\newtheorem{proposition}{Proposition}

\newcommand{\eX}[1]{{\mathbb{E}}[#1]}
\newcommand{\EG}{{\it e.g.}\xspace}
\newcommand{\FI}[1]{Fig.~\ref{#1}\xspace}

\newcommand{\cQ}{\mathcal{Q}}
\newcommand{\dist}{\mathrm{dist}}
\newcommand{\ddist}{\mathrm{dist}}

\newcommand{\G}{\mathcal{G}}

\newcommand{\eqdef}{\stackrel{\mathrm{def}}{=}}
\definecolor{columbiablue}{rgb}{0.61, 0.87, 1.0}

\newcommand{\opt}{{\mathsf{OPT}}}

\newcommand{\emd}{{\sc Emd}}

\newcommand{\aaa}{\mathfrak{A}}
\newcommand{\wdeg}{\mbox{wt-}\deg}

\allowdisplaybreaks

\newtheorem{observation}{Observation}

\journal{Theoretical Computer Science}

\begin{document}

\begin{frontmatter}

\title{On computing Discretized Ricci curvatures of graphs: local algorithms and (localized) fine-grained reductions}

\author[1]{Bhaskar DasGupta\corref{cor1}\fnref{fn1}}
\ead{bdasgup@uic.edu}

\fntext[fn1]{Supported by NSF grant IIS-1814931.}

\cortext[cor1]{Corresponding author.}

\author[2]{Elena Grigorescu\fnref{fn2}}
\ead{elena-g@purdue.edu}

\fntext[fn2]{Supported in part by NSF grants CCF-1910659 and CCF-1910411.}

\author[2]{Tamalika Mukherjee\fnref{fn2}}
\ead{tmukherj@purdue.edu}

\affiliation[1]{organization={Department of Computer Science, University of Illinois Chicago},
            city={Chicago},
            postcode={60607}, 
            state={IL},
            country={USA}}

\affiliation[2]{organization={Department of Computer Science, Purdue University},
            city={West Lafayette},
            postcode={47907}, 
            state={IN},
            country={USA}}


\begin{abstract}
Characterizing shapes of high-dimensional objects via Ricci curvatures plays a critical role in 
many research areas in mathematics and physics. However, even though several discretizations of 
Ricci curvatures for discrete combinatorial objects such as networks have been proposed and 
studied by mathematicians, the computational complexity aspects of these discretizations have 
escaped the attention of theoretical computer scientists to a large extent. In this paper, we 
study one such discretization, namely the Ollivier-Ricci curvature, from the perspective 
of efficient computation by fine-grained reductions and local query-based algorithms.
Our main contributions are the following.
\begin{enumerate}[label=$\triangleright$]
\item
We relate our curvature computation problem to 
minimum weight perfect matching problem on complete bipartite graphs
via fine-grained reduction.
\item
We formalize the computational aspects of the
curvature computation problems in suitable frameworks so that they can be studied by 
researchers in local algorithms. 
\item
We provide the first known lower and upper bounds 
on queries for query-based algorithms for 
the curvature computation problems
in our local algorithms framework. 
\emph{En route}, 
we also illustrate a localized version of our fine-grained reduction.
\end{enumerate}
We believe that our results bring forth an intriguing set of research questions,
motivated both in theory and practice,
regarding designing efficient algorithms for curvatures of geometrical objects.
\end{abstract}

\begin{highlights}
\item 
We relate our curvature computation problem to 
minimum weight perfect matching problem on complete bipartite graphs
via fine-grained reduction.
\item 
We formalize the computational aspects of the
curvature computation problems in suitable frameworks so that they can be studied by 
researchers in local algorithms. 
\item 
We provide the first known lower and upper bounds 
on queries for query-based algorithms for
the curvature computation problems
in our local algorithms framework. 
\emph{En route}, 
we also illustrate a localized version of our fine-grained reduction.
\end{highlights}


\end{frontmatter}



\section{Introduction}

A suitable notion of ``shape'' plays a critical role in investigating objects in mathematics, 
mathematical physics and other research areas.
Various kinds of curvatures are very natural measures of shapes of higher dimensional objects in 
mainstream physics and mathematics~\cite{book99,Berger12}.
To quantify the shape of a higher-dimensional geometric object, one often fixes shapes of objects with specific 
properties as the ``baseline shape'' 
and then quantifies the shape of a given object \emph{with respect to} these baseline shapes.
For example, consider the case of the two-dimensional metric space. For this space, a baseline could be selected as the standard
\emph{Euclidean plane} in which the three angles of a triangle sum up to \emph{exactly} $180^{\circ}$, 
and then one can quantify the shape 
of the given two-dimensional space by the \emph{deviations} of the sum of the three angles of triangles
in this space from the baseline of $180^\circ$.
An alternative approach is to avoid selecting baseline shapes \emph{explicitly} and instead \emph{directly} quantify the shape
of a given geometric object. 
Quantification of shape is often referred to as the \emph{curvature} of the corresponding object.
Quantification of shapes can be either \emph{local} or \emph{global}. 
A local shape of the object is usually computed for a specific \emph{local neighborhood} of the object (\EG, the Ricci curvature).
In contrast, a global shape of the object is usually computed over the entire object (\EG, the Gromov-hyperbolicity measure). 
Any attempt to extend notions of curvature measures from non-network domains to networks\footnote{In this paper, we will use the 
two terms ``graph'' and ``network'' \emph{interchangeably}.}
(and other discrete combinatorial structures)
need to overcome at least three key challenges, namely that 
{(\emph{a})}
networks are \emph{discrete} (non-continuous) combinatorial objects,
{(\emph{b})}
networks may \emph{not} necessarily have an associated natural geometric embedding,
and 
{(\emph{c})}
the extension need to be useful and non-trivial, \IE, 
a network curvature measure
should saliently encode non-trivial higher-order correlations among nodes and edges that
\emph{cannot} be obtained by other popular network measures.

\subsection{Motivations behind studying shapes of networks}

Although studying measures of shapes of networks (and hypergraphs) is mathematically intriguing, it is natural 
to ask if there are other valid reasons for such studies. 
Network shape measures 
\emph{can} encode 
non-trivial topological properties
that are \emph{not} expressed by more 
established network-theoretic measures such as
degree distributions, 
clustering coefficients or 
betweenness centralities (\EG, see~\cite{ADM14,CATAD21}). 
Moreover, 
these shape measures can explain many phenomena one frequently encounters in real network-theoretic applications,
such as 
\textbf{(\emph{i})}
paths mediating up- or down-regulation of a target node starting from the same regulator node in 
\emph{biological regulatory networks} often have many small crosstalk paths~\cite{ADM14} and
\textbf{(\emph{ii})}
existence of congestions in a node that is not a hub in \emph{traffic networks}~\cite{ADM14,JLBB11},
{that are \emph{not} easily explained by other non-shape measures}.
Recently, shape measures have also found applications in traditional social networks applications 
such as community finding~\cite{SJB21}, and in neuroscience applications such as
comparing brain networks to study slowly progressing
brain diseases such as \emph{attention deficit hyperactivity disorder}~\cite{CATAD21} 
and \emph{autism spectrum disorder}~\cite{Ricci,Elumalai2021}.

\subsection{Brief history of existing notions of shapes for networks}

There are several ways previous researchers have attempted to formulate notions of shapes of networks.
Below we discuss \emph{three} major directions in this regard.
For further details and other approaches, the reader is referred to papers and books such 
as~\cite{CL03,book99,Oll11,Oll09,Oll10,Oll07,DJY20,DKMF18,a1,ADM14,CCDDMV18,ipl15,F03,Sree1,Sree2,Weber17,Samal18,CATAD21}.

One notion of network shapes,
first suggested by Gromov in a non-network group theoretic context~\cite{G87}, 
is via the \emph{Gromov-hyperbolicity} of networks.
First defined for infinite continuous metric space~\cite{book99},
the measure was later adopted for finite graphs.
Usually this measure is defined via properties of \emph{geodesic triangles} or equivalently via
$4$-node conditions, though Gromov originally defined the measure using Gromov-product nodes in~\cite{G87}.
Informally, any infinite metric space has a finite Gromov-hyperbolicity measure
if it behaves metrically in the large scale as a \emph{negatively curved} Riemannian manifold, 
and thus the value of this measure can be correlated to the standard scalar curvature of a hyperbolic manifold.
For a finite network the measure is related to the properties of the set of exact and approximate
geodesics of the network.  
There is a \emph{large} body of research works dealing with theoretical and empirical aspects of this measure, \EG, 
see~\cite{DJY20,DKMF18,CCDDMV18,a1,ipl15,CDEHV08}
for theoretical aspects, 
and see~\cite{ADM14,JLBB11,PKBV10}
for 
applications to real-world networks (such as traffic congestions in a road network). 
Gromov-hyperbolicity
is a \emph{global} measure in the sense that it 
assigns one scalar value to the entire network.

A second notion of shape of a network can be obtained by extending Forman's discretization of Ricci curvature for 
(polyhedral or CW) complexes (the ``Forman-Ricci curvature'')~\cite{F03}
to networks. 
Informally, the Forman-Ricci curvature is applied to networks by 
\emph{topologically associating} components (sub-networks) of a given network with higher-dimensional 
objects. 
The topological association itself can be carried out several ways. 
Although formulated relatively recently, there are already a number of 
papers investigating 
properties of 
these measures~\cite{Sree1,Sree2,Weber17,DJY20,Samal18,CATAD21}.

In contrast to both of the above approaches, 
the network curvature considered in this paper 
is obtained via a discretization of curvatures
from Riemannian manifolds to the network domain
to capture metric properties of the manifold that are different from those captured by 
the Forman-Ricci curvature.
More concretely, 
the network curvature studied in this paper is
Ollivier's earth-mover's distances based discretization of Ricci 
curvature (the ``\emph{Ollivier-Ricci curvature}'')~\cite{Oll11,Oll09,Oll10,Oll07}.
For some theoretical comparison 
between 
Ollivier-Ricci curvature and 
Forman-Ricci curvature over graphs, see~\cite{CATAD21}. 

\subsection{Basic definitions and notations}
\label{sec-addl-nota}

Let $G=(V,E)$ be a given undirected unweighted graph.
The following notations related to a graph $G$ 
will be used subsequently: 
\begin{enumerate}[label=$\triangleright$]
\item
$\nbr_G(x)=\{ \, y \,|\, \{x,y\}\in E \}$ 
and 
$\deg_G(x)= |\, \nbr_G(x) \,|$ 
are the set of neighbors and the degree, respectively,  
of a node $x$.
\item
$\ddist_G(x,y)$ 
is the \emph{distance} 
(\IE, number of edges in a shortest path)
between the nodes $x$ and $y$ in $G$.
\end{enumerate}
The following standard notations and terminologies from the field of approximation algorithms 
are used to facilitate further discussions:
\begin{enumerate}[label=$\triangleright$]
\item
$\opt$ is the \emph{value} of the objective of an optimal solution of the problem under discussion. 
\item
A $(\alpha,\eps)$-estimate for a minimization problem under discussion is a polynomial-time algorithm 
that produces a solution whose objective value $\beta$ satisfies 
$\opt \leq \beta \leq \alpha\,\opt+\eps$.
A $(1,\eps)$-estimate is also called an \emph{additive} $\eps$-approximation.
\end{enumerate}

\section{Ollivier-Ricci curvatures: intuition, definitions and simple bounds}
\label{sec-defn-olli}

To define the 
Ollivier-Ricci curvatures
for the components of a graph, 
we first need to 
use the following standard definition of 
the \emph{earth mover's distance} (also called the $L_1$ \emph{Wasserstein distance})
in the specific context of a \emph{edge-weighted complete bipartite} graph.

\begin{definition}[\bf Earth mover's distance (\emd) over a edge-weighted complete bipartite graph]
Let $H=(V_L,V_R,w)$ be an \emph{edge-weighted complete bipartite} graph with $w:V_L\times V_R\mapsto\R^+\cup\{0\}$ 
being the edge-weight function, and let 
$\PP_L:V_L\mapsto\R^+$ and $\PP_R:V_R\mapsto\R^+$ be two arbitrary distributions over the nodes in $V_L$ and $V_R$, respectively. 
The \emph{earth mover's distance} corresponding to the distributions $\PP_L$ and $\PP_R$, 
denoted by \emd$_H(\PP_L,\PP_R)$ $($or simply \emd$)$, 
is the value of the objective function of an optimal solution 
of the following linear program that has a 
variable $z_{x,y}$ for every pair of nodes $x\in V_L$ and $y\in V_R$:
\begin{gather}
\boxed
{
\text{
\begin{tabular}{ r l }
  \emph{minimize} &  
	      $\sum_{x\in V_L}\sum_{y\in V_R} w(x,y)\, z_{x,y}$ 
\\
[5pt]
  \emph{subject to} & 
   $\sum_{y\in V_R} z_{x,y} = \PP_L(x),\,\,$ for all $x\in V_L$
\\
[5pt]
	 & 
   $\sum_{x\in V_L} z_{x,y} = \PP_R(y),\,\,$ for all $y\in V_R$
\\
[5pt]
	& 
  $z_{x,y}\geq 0,\,\,$ for all $x\in V_L$ and $y\in V_R$
\end{tabular}
}
}
\label{lp1}
\end{gather}
\end{definition}

Let $G=(V,E)$ be an undirected unweighted graph.
Consider 
an edge $e=\{u,v\}\in E$.
Define the \textbf{edge-weighted complete bipartite} graph $G_{u,v}=(L_{u,v}^G,R_{u,v}^G,w_{u,v}^G)$ as follows: 
\begin{enumerate}[label=$\triangleright$]
\item
$L_{u,v}^G =\{u\}\cup \nbr_G(u)$, 
\item
$R_{u,v}^G =\{v\}\cup \nbr_G(v)$, and 
\item
the edge-weight function $w_{u,v}^G$ is given by 
$w_{u,v}^G(u',v')=\ddist_G(u',v')$ for all $u'\in L_{u,v}^G,v'\in R_{u,v}^G$.
\end{enumerate}
Let 
$\PP_u^G$ and $\PP_v^G$
denote the two uniform distributions over the nodes in 
$L_{u,v}^G$ and $R_{u,v}^G$, respectively, \IE,  
\begin{gather*}
\forall \, x\in L_{u,v}^G: \,
\PP_u^G(x)= \frac{1}{1+\deg_G(u)} 
\\
\forall \, x\in R_{u,v}^G: \,
\PP_v^G(x)= \frac{1}{1+\deg_G(v)} 
\end{gather*}
We can now state the precise definitions of the curvatures used in this paper. 
\begin{enumerate}[label=$\triangleright$]
\item
The \textbf{Ollivier-Ricci curvature of the edge $\pmb{e=\{u,v\}}$ of $G$} is 
defined as~\cite{Oll11}\footnote{For this paper, it is crucial to note that the computation of 
$\mathfrak{C}_{G}(e)$ requires \emph{only} the value of 
\emd$_{G_{u,v}}(\PP_u^G,\PP_v^G)$ and does \emph{not} require an explicit enumeration of 
the solution (variable values) of the linear program~\eqref{lp1}.
This distinction is important in the context of designing efficient local algorithms.
For example, 
given a graph $G$ with $n$ nodes in which the maximum degree of any node is $O(1)$ and a constant $\eps>0$, one can 
compute a number that is an additive $\eps n$-approximation of the size of maximum matching of $G$ in 
$O(1)$ time in expectation~\cite{YYI2012}, but of course if we were required to output an actual maximum matching we would take
at least $\Omega(n)$ time.}
\begin{gather}
\boxed
{
\mathfrak{C}_{G}(e) \eqdef
\mathfrak{C}_{G}(u,v) = 
1 - \text{\emd}_{\!\!\!G_{u,v}}(\PP_u^G,\PP_v^G)
}
\label{eq-def-ricci-edge}
\end{gather}
\item
The \textbf{Ollivier-Ricci curvature of a node} $\pmb{v}$ 
is calculated by taking the average of the 
Ollivier-Ricci curvatures of all the edges incident on $v$, \IE, 
\begin{gather}
\boxed
{
\mathfrak{C}_{G}(v) = \frac{1}{\deg_G(v)} 
\sum_{e=\{u,v\}\in E} \mathfrak{C}_{G}(e) 
}
\label{eq-def-ricci-node}
\end{gather}
\item
Finally, the \textbf{average Ollivier-Ricci curvature of a graph} $\pmb{G}$ 
is calculated by taking the average of the 
Ollivier-Ricci curvatures of all the edges in $G$, \IE, 
\begin{gather}
\boxed
{
\mathfrak{C}_{\mathrm{avg}}(G) = \frac{1}{|E|} \sum_{e\in E} \mathfrak{C}_{G}(e) 
}
\label{eq-def-ricci-ave}
\end{gather}
\end{enumerate}
For easy quick reference, we explicitly write below the version of  
the linear program in~\eqref{lp1}
as used in the calculation of 
$\mathfrak{C}_{G}(u,v)$:
\begin{gather}
\hspace*{-0.2in}
{
\text{
\begin{tabular}{ | l | }
\hline
  \emph{minimize} 
	      $\sum_{x\in \{u\}\cup \nbr_G(u)}\sum_{y\in \{v\}\cup \nbr_G(v) } \dist_G(x,y)\, z_{x,y}$ 
\\
[5pt]
  \emph{subject to} 
\\
[5pt]
\hspace*{0.2in}
   $\sum_{y\in \{v\}\cup \nbr_G(v)} z_{x,y} = \frac{1}{1+\deg_G(u)},\,\,$ for all $x\in \{u\}\cup \nbr_G(u)$
\\
[5pt]
\hspace*{0.2in}
   $\sum_{x\in \{u\}\cup \nbr_G(u)} z_{x,y} = \frac{1}{1+\deg_G(u)},\,\,$ for all $y\in \{v\}\cup \nbr_G(v)$
\\
[5pt]
\hspace*{0.2in}
  $z_{x,y}\geq 0,\,\,$ for all $x\in \{v\}\cup \nbr_G(v)$ and $y\in \{u\}\cup \nbr_G(u)$
\\
\hline
\end{tabular}
}
}
\tag{$\LP$-$\mathfrak{C}_{G}$}
\label{lpcg}
\end{gather}
Assuming 
$\deg_G(u)\leq \deg_G(v)$,
the linear program in~\eqref{lpcg}
has 
$\deg_G(u)\allowbreak \times \allowbreak \deg_G(v)\allowbreak =\allowbreak O( (\deg_G(v) )^2 )$ variables and 
$\deg_G(u)+\deg_G(v)\leq 2\,\deg_G(v)$ constraints. 
The best time-complexity for solving the 
linear program in~\eqref{lpcg}
can be estimated as follows:
\begin{enumerate}[label=$\triangleright$]
\item
Based on the state-of-the-art algorithms for solving linear program for this situation~\cite{LS15}, 
an exact solution of \eqref{lpcg} can be found in 
$O( (\deg_G(v) )^{5/2} )$ time.
\item
Based on the results in publications such as~\cite{Q18,DGK18}, an additive 
$\eps$-approximation of \eqref{lpcg} can be obtained in 
$\tilde{O} \big( \frac{1}{\eps^{2}}\deg_G(u)\deg_G(v) \big)
=
\tilde{O} \big( \frac{1}{\eps^{2}} (\deg_G(v))^2 \, \big)
$ 
time\footnote{The standard $\tilde{O}$ notation in algorithmic analysis hides poly-logarithmic terms, 
\EG, terms like $\log^{4/3}\deg_G(v)$.}.
\end{enumerate}

The following observation is crucial for this paper.

\begin{observation}\label{obs3}
The values $\dist_G(x,y)$ in 
the linear program in~\eqref{lpcg}
satisfy the property that 
$\dist_G(x,y)\in\{0,1,2,3\}$.
\end{observation}

It is not difficult to see that Observation~\ref{obs3} implies 
$0\leq \text{\emd}_{\!\!\!G_{u,v}}(\PP_u^G,\PP_v^G)\leq 3$
and therefore 
$-2\leq \mathfrak{C}_{G}(e) \leq 1$.
For computing 
$\mathfrak{C}_{G}(u,v)$ and related quantities,
\textbf{we assume that $\pmb{\deg_G(u)\leq \deg_G(v)}$
without any loss of generality
throughout the rest of the paper}.
Moreover, 
\textbf{we also assume
without loss of generality
that 
$\pmb{\deg_G(v)=\omega(1)}$ since otherwise 
$\pmb{\text{\emd}_{\!\!\!G_{u,v}}(\PP_u^G,\PP_v^G)}$
can be computed in $\pmb{O(1)}$ time}.

\subsection{Intuition behind the discretization resulting in definition of $\mathfrak{C}_{G}(e)$} 

For an intuitive understanding of the
definition of 
$\mathfrak{C}_{G}(e)$, 
we recall the notion of Ricci curvature for 
a smooth Riemannian manifold.
The Ricci curvature at a point $x$ in the manifold along a direction can be thought of 
transporting a small ball centered at $x$ along that direction and measuring the ``distortion'' 
of that ball due to the shape of the surface by comparing the distance between the two small balls with
the distance between their centers.
In the definition of $\mathfrak{C}_{G}(e)$, 
the role of the direction is captured by the edge $e=\{u,v\}$,
the roles of the balls at the two points are played by the two 
closed neighborhoods 
$L_{u,v}^G$
and 
$R_{u,v}^G$, 
and the role of the 
distance between the two balls 
is captured by the 
earth mover's distance between the two distributions 
$\PP_u^G$ and $\PP_v^G$ 
over the nodes in 
$L_{u,v}^G$ and $R_{u,v}^G$
on the metric space of shortest paths in $G$.
For further intuition, see publications such as~\cite{Oll11}.
The Forman-Ricci curvature also assigns a number to each edge of the given graph, 
but the numbers are calculated in quite a different way 
from that in the Ollivier-Ricci curvature to capture different metric properties of the manifold.

\subsection{Equivalent reformulation of linear program~\eqref{lpcg} when $\deg_G(u)=\deg_G(v)$}
\label{sec-eqdeg}

The following claim holds based on 
results in prior publications such as~\cite{ASD20,PC19}.
For the convenience of the reader, we provide a self-contained proof in the appendix.

\begin{fact}{\emph{\cite{ASD20,PC19}}}\label{fact2}
If  $\deg_G(u)=\deg_G(v)$ then
the following claims are true regarding some optimal solution of 
the linear program~\eqref{lpcg}:
\begin{enumerate}[label=\textbf{\emph{(}\roman*\emph{)}}]
\item
The values of the variables $z_{u',v'}$ are either $0$ or $\frac{1}{\deg_G(v)}$.
\item
The edges in 
$\left\{ \,\{u',v'\} \,|\, z_{u',v'}= \frac{1}{\deg_G(v)} \,\right\}$ 
form a minimum-weight perfect matching in $G_{u,v}$ that uses the zero-weight edges 
$\{u',u'\}$
for all 
$u'\in \{u,v\} \cup \big( \,\nbr_G(u)\cap \nbr_G(v)  \,\big)$.
\end{enumerate}
\end{fact}

Based on Fact~\ref{fact2}, for the case when when $\deg_G(u)=\deg_G(v)$
an optimal solution of the  
linear program~\eqref{lpcg}
can be obtained by 
finding a \emph{minimum-weight perfect matching} 
for a \emph{complete edge-weighted bipartite} graph 
$H=(L,R,w)$ where 
\begin{enumerate}[label=$\triangleright$]
\item
$L=\nbr_G(u)\setminus \big( \nbr_G(v) \cup \{v\} \big)$,
\item
$R=\nbr_G(v)\setminus \big( \nbr_G(u) \cup \{u\} \big)$,
and
\item
the edge-weight function
$w:L\times R \mapsto \{1,2,3\}$ is 
given by $w(x,y)=\dist_G(x,y)$.
\end{enumerate}
Note that $|L|=|R|=\deg_G(v)-1-|\nbr_G(u)\cap\nbr_G(v)|$.
Letting $\cM(H) \in \{|R|,|R|+1,\dots,3\,|R|\}$ denote the total weight of a minimum-weight perfect matching of $H$, 
we have 
\begin{gather*}
\mathfrak{C}_{G}(e)=
1 - \frac{\cM(H)}{1+\deg_G(v)}
\end{gather*}

\begin{proposition}\label{obs1}
An additive $\eps |R|$-approximation of $\cM(H)$ 
implies 
an additive $\eps$-approximation of 
$\mathfrak{C}_{G}(e)$,
and vice versa.
\end{proposition}

\begin{proof}
This follows from the facts that 
$|R|\leq \cM(H)\leq 3\,|R|$
and 
$
\deg_G(v) > |R|
$.
\end{proof}

\subsection{Some simple bounds for 
$\text{\emd}_{\!\!\!G_{u,v}}(\PP_u,\PP_v)$
and 
$\mathfrak{C}_{G}(e)$}
\label{sec-simple-bo}

We use a calculation similar to the one used in~\cite{ASD20}.
Extend the distributions 
$\PP_u^G$ and $\PP_v^G$
to 
$\PP_u^{G'}$ and $\PP_v^{G'}$ over 
$L_{u,v}^G\cup R_{u,v}^G$ 
by letting 
$\PP_u^{G'}(x)=0$ for $x\in R_{u,v}\setminus L_{u,v}$ and 
$\PP_v^{G'}(x)=0$ for $x\in L_{u,v}\setminus R_{u,v}$. 
For notational simplicity, let 
$k=\nbr_G(u)\setminus \nbr_G(v)$, 
$\ell=\nbr_G(u)\cap\nbr_G(v)$, 
and
$m=\nbr_G(v)\setminus \nbr_G(u)$, 
thus 
$\deg_G(u)=k+\ell$
and 
$\deg_G(v)=m+\ell$.
By straightforward calculation, the total variation distance (TVD) 
between 
$\PP_u'$ and $\PP_v'$
is
\begin{multline*}
\textstyle
|| \, \PP_u'-\PP_v' \, || _{ \mathrm{TVD} }
=
\frac{1}{2} \times 
\left(
\frac{k-1}{k+\ell+1}
+
\frac{m-1}{m+\ell+1}
+ 
(\ell+2) \times 
\left( 
\frac{1}{k+\ell+1}
-
\frac{1}{m+\ell+1}
\right)
\right)
\\
\textstyle
=
1 
-
\frac{\ell+2}{\deg(v)+1}
\end{multline*}
Since 
$1\leq \ddist_G(u',v')\leq 3$ for all $u',v'\in L_{u,v}^G\cup R_{u,v}^G, u'\neq v'$, 
by standard relationships between \emd\ and TVD (\EG, see~\cite{GS02}) it follows that 
$
|| \, \PP_u'-\PP_v' \, || _{ \mathrm{TVD} }
\leq 
\text{\emd}_{\!\!\!G_{u,v}}(\PP_u,\PP_v)
\leq 
3 \times || \, \PP_u'-\PP_v' \, || _{ \mathrm{TVD} }
$, thereby giving 
\begin{gather*}
-2 + 
\frac{3\ell+6}{\deg_G(v)+1}
\leq 
\mathfrak{C}_{G}(e)
\leq 
\frac{\ell+2}{\deg_G(v)+1}
\end{gather*}
Furthermore, if $G$ has no cycles of length $5$ or less containing $e$ then 
$\ddist_G(u',v')=3$ for all $u',v'\in L_{u,v}^G\cup R_{u,v}^G$ and $\ell=0$ giving 
$
\mathfrak{C}_{G}(e)
= 
\frac{2}{\deg_G(v)+1}
$.

\newcommand{\mmid}{{\mathfrak{m}\mathfrak{i}\mathfrak{d}}}

\section{Synopsis of our results}

The main goal of this paper is to study algorithmic complexities of 
efficient computation of our network curvature measures. 
To this effect, our main contributions are threefold:
\begin{enumerate}[label=$\triangleright$]
\item
We relate various cases of our curvature computation problems via fine-grained reduction.
\item
We formalize the computational aspects of the
curvature computation problems in suitable frameworks so that they can be studied by 
researchers in local algorithms. 
\item
We provide the first known lower and upper bounds 
on queries for query-based algorithms for 
the curvature computation problems
in our local algorithms framework. 
\emph{En route}, 
we also illustrate a localized version of our fine-grained reduction.
\end{enumerate}
A summary of our contribution 
in the rest of this paper is the following.
\begin{enumerate}[label={\small\Pisymbol{pzd}{113}},leftmargin=*]
\item
In Section~\ref{sec-related-eq-noteq}
we relate  
via Theorem~\ref{thm-ineqdeg}
the minimum weight perfect matching problem on complete bipartite graphs with ternary weights 
to computing $\mathfrak{C}_{G}(e)$ via fine-grained reduction.
\item
In Section~\ref{sec-query-all}
we present our results 
for computing $\mathfrak{C}_{G}(e)$ in the framework of local algorithms.
\begin{enumerate}[label={\Large$\circ$},leftmargin=*]
\item
In Sections~\ref{sec-prior-query}$\,$--$\,$\ref{sec-query-defn}
we provide details of the query models relevant to our case and prior related works on these query
models.
\item
In Sections~\ref{suc-sub-lower}$\,$--$\,$\ref{suc-sub-upper2}
Theorem~\ref{thm-lower-all},
Theorem~\ref{thm-simple-up}
and 
Theorem~\ref{thm-trans-query}
provide
query bounds for exact or approximate calculations of the curvature 
using the query models.
The bounds are succinctly summarized in 
Section~\ref{sec-query-bounds-all} via
Table~\ref{tabb1}.
\end{enumerate}
\item
In Section~\ref{sec-ricci-nodes} 
Lemma~\ref{lem-ave}
provides our results for computing 
the Ollivier-Ricci curvature $\mathfrak{C}_{G}(v)$ 
for nodes and 
for computing the average Ollivier-Ricci curvature 
$\mathfrak{C}_{\mathrm{avg}}(G)$
for graphs using 
``black box'' additive approximation algorithms for $\mathfrak{C}_{G}(e)$
and 
neighbor queries.
\item
We conclude in Section~\ref{sec-conclude}
with some possible future research problems.
\end{enumerate}

\newcommand{\mpmct}{{\sc Mpmct}}

\section{Fine-grained reduction: relating minimum weight perfect matching on complete bipartite graphs
to computing $\mathfrak{C}_{G}(e)$}
\label{sec-related-eq-noteq}

Frameworks for characterizing polynomial-time solvable problems via \emph{fine-grained reduction} 
have garnered considerable attention in recent years (\EG, see~\cite{VVW19} for a survey and~\cite{AGW15,Pat10,Lee02} 
for a few well-known results in this direction). 
Essentially these fine-grained reductions are used to show that, 
given two problems $\cA$ and $\cB$ and two constants $a,b>0$, 
if an instance $\cI_{\cB}$ of size $|\cI_{\cB}|$ of problem $\cB$ can be solved in $O(|\cI_{\cB}|^b)$ time  
then an instance $\cI_{\cA}$ of size $|\cI_{\cA}|$ of problem $\cA$ can be solved in $O(|\cI_{\cA}|^a)$ time.

To begin, we first formally state the \emph{minimum weight perfect} matching problem on \emph{complete bipartite} graphs with \emph{ternary} 
edge weights. 

\begin{definition}[\bf minimum weight perfect matching on complete bipartite graphs with ternary weights (\mpmct)]
Given a complete edge-weighted bipartite graph $H=(A,B,w)$ where $|A|=|B|$ and 
$w:A\times B \mapsto \{1,2,3\}$ is the edge-weight function, find 
the value of 
$\frac{|\cM|}{|A|}$ where $|\cM|$ is the value $($sum of weights of edges$)$ in 
a minimum-weight perfect matching $\cM$ of $H$.
\end{definition}

For \mpmct, exact solution takes $O( |A|^{5/2} )$ time~\cite{LS15}, 
and an $\eps$-additive approximation takes 
$\tilde{O} \big( \frac{1}{\eps^{2}} |A|^2 \, \big)$ time.
The following theorem related \mpmct\ to the problem of computing a solution of 
the linear program in~\eqref{lpcg} via a fine-grained reduction.

\begin{theorem}\label{thm-ineqdeg}
Suppose that we have an algorithm
$\aaa$ that provides 
$(\alpha,\eps)$-estimate for 
\mpmct\
in $O(|A|^{2+\mu})$ time for some $\mu\geq 0$ for a given input instance $H=(A,B,w)$.

Then, there exists an algorithm 
$\aaa_<$ that provides
the following estimates for the linear program in~\eqref{lpcg}
in $O(\deg_G(v)^{2+\mu})$ time:
\begin{enumerate}[label={\emph{(}\roman*\emph{)}}]
\item
$(\alpha,\eps)$-estimate if 
$\deg_G(v)+1$ is an integral multiple of $\deg_G(u)+1$, and 
\item
$(\alpha,\eps+\delta)$-estimate 
$($for $\delta>0)$
provided $\delta$ satisfies 
at least one of the following conditions:
\begin{enumerate}[label={\emph{(}\alph*\emph{)}}]
\item
$\deg_G(u)\leq(\delta/3)\times \deg_G(v)$,
or 
\item
$\deg_G(u)\geq (1 - (\delta/3)\,) \times \deg_G(v)$.
\end{enumerate}
\end{enumerate}
\end{theorem}

\begin{remark}
An illustration of the result in 
Theorem~{\em\ref{thm-ineqdeg}} is as follows.
Suppose that we can solve \mpmct\ exactly in 
$O( |A|^{2.4} )$ time $($implying $\alpha=1$ and $\eps=0)$. 
Then, such an algorithm can be used to obtain a 
${\deg_G(v)}^{-1/2}$-additive approximation of~\eqref{lpcg} 
\emph{(\IE}, $\delta={\deg_G(v)}^{-1/2})$
in $O( (\deg_G(v) )^{2.4} )$ time 
provided 
at least one of the following conditions hold:
\begin{enumerate*}[label={\emph{(}\alph*\emph{)}}]
\item
$\deg_G(u)\leq \frac{\sqrt{\deg_G(v)}}{3}$,
\item
$\deg_G(u)\geq \deg_G(v) - \frac{\sqrt{\deg_G(v)}}{3}$, 
or
\item
$\deg_G(v)+1$ is an integral multiple of $\deg_G(u)+1$. 
\end{enumerate*}
Such a result will improve the best possible running time for a 
${\deg_G(v)}^{-1/2}$-additive approximation of~\eqref{lpcg}. 
\end{remark}

\begin{proof}
Let 
$\nbr_G(u)\cup \{u\}=\{x_1,\dots,x_{\deg(u)+1}\}$, 
and 
$\nbr_G(v) \cup \{v\}=\{y_1,\dots,y_{\deg(v)+1}\}$, 
where $x_{\deg_G(u)}=y_{\deg_G(v)}=u$ and 
$x_{\deg_G(u)+1}=y_{\deg_G(v)+1}=v$.
Let 
$\deg_G(v)+1=a\,(\deg_G(u)+1)+b$ for two integers $a\geq 1$ and $0\leq b<\deg_G(u)+1$. 
We construct a new 
graph $G_{u,v}'=(L_{u,v}^{G'},R_{u,v}^{G'},w_{u,v}^{G'})$ 
from $G_{u,v}$
in the following manner: 
\begin{enumerate}[label=$\triangleright$]
\item
We set $R_{u,v}^{G'}=R_{u,v}^{G}$.
\item
Every node 
$x_i$
is replaced by $a$ nodes $x_i^1,\dots,x_i^a$
in $L_{u,v}^{G'}$.
Moreover, we have $b$ additional ``special'' nodes  
$r_1,\dots,r_b$
in $L_{u,v}^{G'}$.
Note that after these modifications 
$| L_{u,v}^{G'}| =| R_{u,v}^{G'}| = 1+\deg_G(v)$.
\item
We set 
the new weights $w_{u,v}^{G'}$
as follows: 
\begin{eqnarray*}
w_{u,v}^{G'}
(x_i^j,y_\ell) 
      & =  & \ddist_{G}(x_i,y_\ell)
     \text{ for $i\in\{1,\dots,\deg_G(u)+1\}$,}
\\
		 & & \hspace*{0.5in} \text{$j\in\{1,\dots,a\}$, and $\ell\in\{1,\dots,\deg_G(v)+1\}$}
\\
w_{u,v}^{G'}
(r_i,y_\ell) 
      & =  & 3 \text{ for $i\in\{1,\dots,b\}$, and $\ell\in\{1,\dots,\deg_G(v)+1\}$}
\end{eqnarray*}
\item
The two \emph{new} probability distributions 
${\PP_u}^{\!\!\!\!G_{u,v}'}$ and ${\PP_v}^{\!\!\!\!G_{u,v}'}$
over the nodes in $L_{u,v}^{G'}$ and $R_{u,v}^{G'}$ are 
as follows: 
${\PP_v}^{\!\!\!\!G_{u,v}'}(x)=\PP_v^G(x)$
for all $x\in\{y_1,\dots,y_{\deg_G(v)+1}\}$,
and 
$
{\PP_u}^{\!\!\!\!G_{u,v}'}(x)
= 
\frac{1}{1+\deg_G(v)}$ for all $x\in \bigcup_{i,j} \{ x_i^j\} \, \cup \{r_1,\dots,r_b\}$.
\end{enumerate}
Since 
$| L_{u,v}^{G'}| =| R_{u,v}^{G'}| = 1+\deg_G(v)$,
using the reformulations as discussed in Section~\ref{sec-eqdeg}
it follows that 
$G_{u,v}'$ 
is a valid instance $H=(A,B,w)$ of 
\mpmct\ with $|A|=\deg_G(v)+1$ and $w(p,q)=w_{u,v}^{G'}(p,q)$.
Note that building 
the graph $G_{u,v}'$ takes 
$O( ( \deg_G(v) )^2 )$ time,
and algorithm 
$\aaa$ provides a 
$(\alpha,\eps)$-estimate for 
$\text{\emd}_{\!\!G_{u,v}'}( {\PP_u}^{\!\!\!\!G_{u,v}'}, {\PP_v}^{\!\!\!\!G_{u,v}'})$ 
in $O( ( \deg_G(v) )^{2+\mu} )$ time.
Thus, 
to complete the proof
it suffices to show that 
\[
\text{\emd}_{\!\!\!G_{u,v}} (\PP_u^G,\PP_v^G)
\leq
\text{\emd}_{\!\!G_{u,v}'}( {\PP_u}^{\!\!\!\!G_{u,v}'}, {\PP_v}^{\!\!\!\!G_{u,v}'})
\leq
\text{\emd}_{\!\!\!G_{u,v}}
(\PP_u^G,\PP_v^G)
    + \delta
\] 
The linear program for 
$\text{\emd}_{\!\!G_{u,v}'}( {\PP_u}^{\!\!\!\!G_{u,v}'}, {\PP_v}^{\!\!\!\!G_{u,v}'} )$
is a straightforward modified version of~\eqref{lpcg} with appropriate change of subscripts of the variables.
We will refer to this modified version by~\eqref{lpcg}$'$.

We can show 
$
\text{\emd}_{\!\!\!G_{u,v}'}( {\PP_u}^{\!\!\!\!G_{u,v}'}, {\PP_v}^{\!\!\!\!G_{u,v}'})
\leq
\text{\emd}_{\!\!\!G_{u,v}} (\PP_u^G,\PP_v^G)
   + \delta
$ 
as follows.
Consider an \emph{optimal} solution of the linear program~\eqref{lpcg} 
of value 
$\text{\emd}_{\!\!\!G_{u,v}} (\PP_u^G,\PP_v^G)$.
From this solution we can create a \emph{feasible solution} of 
the linear program~\eqref{lpcg}$'$
in the following manner.
\begin{enumerate}[label=$\triangleright$]
\item
For $i=1,\dots,\deg_G(u)+1$ and $j=1,\dots,\deg_G(v)+1$, 
if $z_{x_i,y_j}>0$ then  
distribute the value of $z_{x_i,y_j}$ 
among the corresponding variables of~\eqref{lpcg}$'$ as follows:
\begin{itemize}
\item
Repeatedly select a variable from 
$\{x_i^1,\dots,x_i^a\}$, say $x_i^\ell$, such that $x_i^\ell<\frac{1}{1+\deg_G(v)}$. 
Increase $x_i^\ell$ to 
$\min \left\{ \frac{1}{(1+\deg_G(v))},\, z_{x_i,y_j} \right\}$, and 
decrease $z_{x_i,y_j}$ by the amount by which $x_i^\ell$ was increased. 
Note that 
$
w_{u,v}^{G'}(x_i,y_j)=
\ddist_{G}(x_i,y_j)
$.
Repeat this step until $z_{x_i,y_j}$ becomes zero or no such variable 
$x_i^\ell$ exists.
\item
If 
$z_{x_i,y_j}>0$
after the previous step ends 
then 
execute the following steps.
Repeatedly select a variable from 
$\{r_1,\dots,r_b\}$, say $r_\ell$, such that $r_\ell<\frac{1}{1+\deg_G(v)}$. 
Increase $r_\ell$ to 
$\min \left\{ \frac{1}{(1+\deg_G(v))},\, z_{x_i,y_j} \right\}$, and 
decrease $z_{x_i,y_j}$ by the amount by which $r_\ell$ was increased. 
Note that 
$
w_{u,v}^{G'}(x_i,y_j) \leq
\ddist_{G}(x_i,y_j)+3
$.
Repeat this step until $z_{x_i,y_j}$ becomes zero.
\end{itemize}
\end{enumerate}
A straightforward calculation show that 
$
\text{\emd}_{\!\!\!G_{u,v}'}( {\PP_u}^{\!\!\!\!G_{u,v}'}, {\PP_v}^{\!\!\!\!G_{u,v}'})
\leq
\text{\emd}_{\!\!\!G_{u,v}} (\PP_u^G,\PP_v^G)
     + \frac{3b}{\deg_G(v)+1}
$.
Therefore it suffices if we have 
$\frac{3b}{\deg(v)+1}\leq \delta$.
If $\deg_G(u)+1$ is an integral multiple of $\deg_G(v)+1$ then 
$b=0$ and this proves the claim in 
(\emph{i}).
Otherwise, 
since $b<\deg_G(u)+1\leq \deg_G(v)+1$ and 
$b\leq (\deg_G(v)+1)-(\deg_G(u)+1)= \deg_G(v)-\deg_G(u)$ 
we get 
\begin{multline*}
\textstyle
\deg_G(u)\leq(\delta/3)\times \deg_G(v)
\,\Rightarrow\,
b<\deg_G(u)+1 \leq (\delta/3)\times \deg_G(v) +1
\\
\textstyle
\Rightarrow\,
\frac{3b}{\deg(v)+1} < \frac{\delta\times \deg_G(v) +1 }{\deg(v)+1}
\leq \delta
\end{multline*}
\begin{multline*}
\textstyle
\deg_G(u)\geq (1 - (\delta/3)\,) \times \deg_G(v)
\,\Rightarrow\,
\deg_G(v)-\deg_G(u)  \leq (\delta/3) \times \deg_G(v)
\\
\textstyle
\Rightarrow\,
\frac{3b}{\deg(v)+1}  \leq 
\frac{  \delta\times\deg_G(v) }{\deg(v)+1}  <\delta
\end{multline*}
The proof of 
$
\text{\emd}_{\!\!\!G_{u,v}} (\PP_u^G,\PP_v^G)
\leq
\text{\emd}_{\!\!G_{u,v}'}( {\PP_u}^{\!\!\!\!G_{u,v}'}, {\PP_v}^{\!\!\!\!G_{u,v}'})
$
is similar.
\end{proof}

\section{Computing $\mathfrak{C}_{G}(e)$ in the framework of local algorithms}
\label{sec-query-all}

By now designing local algorithms for efficient solution of graph-theoretic problems 
has become a well-established research area in theoretical computer science and data mining 
with a large body of publications (\EG, see~\cite{PARNAS2007,YYI2012,ORR12}).
A basic idea behind many of these algorithms is to suitably sample a small ``local'' neighborhood
of the graph to infer the value of some non-local property of a graph.
Frameworks for graph-theoretic applications of local algorithms 
hinges on the following two premises:
\begin{enumerate}[label=$\triangleright$]
\item
We assume that our algorithm has a list of all nodes in the graph in a suitable format that 
allows for sampling a node based on some distribution.
\item
The edges and their weights are \emph{not} known to our algorithm \emph{a priori}. 
Instead, the algorithm uses a ``query'' on a node or a pair of nodes to discover an edge and its weight.
Different query models for local algorithms arise based on what kind of queries are allowed.
Later in Section~\ref{sec-query-defn}
we will provide details of query models that are applicable to our problems.
\item
The performance of the algorithm is measured by the \emph{number} of queries used.
\end{enumerate}

\noindent
\textbf{Additional notations and conventions} 
\medskip

For the case when $\deg_G(u)=\deg_G(v)$, 
we will use the reformulations of the linear program~\eqref{lpcg} as discussed in Section~\ref{sec-eqdeg}
and the associated notations contained therein.
We will use the following additional notations and conventions related to the graph $H=(L,R,w)$
mentioned in  Section~\ref{sec-eqdeg}:
\begin{enumerate}[label=$\triangleright$]
\item
$|L|=|R|=n$, $L=\{u_1,\dots,u_n\}$ and $R=\{v_1,\dots,v_n\}$.
\item
For $j\in\{1,2\}$ $\deg_{H,j}(x)$ denotes the number of edges of weight $j$ incident on 
node $x$ in a graph $H$. 
\end{enumerate}
\textbf{Note that $\pmb{H}$ has $\pmb{2n}$ nodes}.
Moreover, 
for any edge-weighted graph $F=(V,E,w)$ with $w:E\mapsto \R$, 
$\wdeg_F(u)=\sum_{v:\{u,v\}\in E} w(u,v)$
denotes the \emph{weighted degree} of node $u$ in $F$, and 
$\cM(F)$ denotes the total weight of a minimum-weight perfect matching of $F$. 

\subsection{Prior related works} 
\label{sec-prior-query}

Designing sublinear time and sketching algorithms for the general earth mover's distance 
on the shortest path metric for arbitrary graphs have been investigated in prior research papers such as~\cite{ba2009sublinear,MS13}.
In particular, for an edge-weighted tree with $W$ being the maximum weight of any edge and for any two unknown probability
distributions on the nodes, 
the authors in~\cite{ba2009sublinear}
show that an estimate of the \emd\ with $\eps$-additive error 
can be achieved by using 
$\tilde{O}(\frac{W^2 n^2}{\eps^2})$
samples from the two distributions and observes that their algorithm is optimal up to polylog factors.
To the best of our knowledge, local algorithms for computing 
the Ollivier-Ricci curvatures
of a graph have \emph{not} been investigated explicitly before.

\subsection{Query models for edge-weighted complete bipartite graphs}
\label{sec-query-defn}

Two standard query models that appear in the local algorithms literature
for unweighted graphs~(\EG, see~\cite{PARNAS2007}) are as follows: 
the \emph{node-pair} query model (the query is a pair of nodes and the answer is whether an edge between them exists or not), 
and 
the \emph{neighbor} query model (the query is a node and the answer is a random not-yet-explored adjacent node if it exists).
Since our given graph is an edge-weighted complete bipartite graph $H=(L,R,w)$ via the reformulation described in 
Section~\ref{sec-eqdeg},
natural extensions lead to the following query models for our case: 
\begin{enumerate}[label=$\triangleright$]
\item
\textbf{weighted node-pair query model}:
the query is a pair of nodes $x,y$ and the answer is 
the \emph{weight} $w(x,y)$.
\item
\textbf{neighbor query model}:
the query is a node $x$ and the answer is a random ``not-yet-explored'' node adjacent to $x$ (if no such
node exists then the query returns a special symbol to indicate that).
{\em Note that such a query does not give any useful information $($beyond simply picking a node uniformly at random$)$ 
for the graph $H$ since it is a complete graph}.
We will only use this type of query for the entire given graph $G$ 
for computing $\mathfrak{C}_{G}(v)$ and $\mathfrak{C}_{\mathrm{avg}}(G)$
in Section~\ref{sec-ricci-nodes}.
\item
\textbf{weighted neighbor query model}:
the query is $(x,y)$ where $x$ is a node and $y$ is a number,
and the answer is a random ``not-yet-explored'' node $z$ such that $w(x,z)=y$ (if no such node exists
then the query returns a special symbol to indicate that).
\item
\textbf{weighted selective degree query model}:
the query is $(x,y)$ where $x$ is a node and $y$ is a number,
and the answer is the number of edges of weight $y$ that are incident on $x$. 
\end{enumerate}

\subsection{Summary of our query bounds on computing $\mathfrak{C}_{G}(e)$}
\label{sec-query-bounds-all}

For the convenience of the reader, we summarize our query bounds 
for computing $\mathfrak{C}_{G}(e)$
in Table~\ref{tabb1}.
Subsequent sub-sections in this section provide proofs of these bounds.

\begin{table}[htbp]
\scalebox{0.8}
{
\begin{tabular}{c c c c l c}
\hline
& \multicolumn{1}{c|}{query} &  additive & expected & \multicolumn{1}{c}{result(s)}  & additional 
\\
& \multicolumn{1}{c|}{types} & approx. & \# of queries & & remark(s)
\\
\cmidrule{2-6}
\multirow{5}{*}{\makecell{lower \\ bounds}}
	\ldelim\{{5}{*}
&
weighted node-pair  & 
			\makecell{
        exact 
				\\
				computation 
				}
	& $>\frac{(\deg_G(v)-1)^2}{6}$
	& 
     Theorem~\ref{thm-lower-all}(\emph{a})-(\emph{i})
	&
	\rdelim ]{5}{*}[{\Large\ding{172}}]
\\
\cmidrule{2-5}
& weighted neighbor 
     & 
			\makecell{
        exact 
				\\
				computation 
       }
		 & $>\frac{\deg_G(v)-1}{2}$
	& 
     Theorem~\ref{thm-lower-all}(\emph{a})-(\emph{ii})
	&
\\
\cmidrule{2-5}
& weighted node-pair  & $2-\eps_1$ &   $>\frac{\deg_G(v)-1}{6}$
   & 
     Theorem~\ref{thm-lower-all}(\emph{b})
	 &
\\
\midrule
\multirow{7}{*}{\makecell{upper \\ bounds}}
	\ldelim\{{7}{*}
& weighted neighbor  & $1+\eps_2$ &   $O(1)$ 
   & 
	 Theorem~\ref{thm-simple-up}(\emph{a})
	 & 
   {\Large\ding{173}}
\\
\cmidrule{2-6}
& weighted neighbor  & $\frac{1}{2}+\eps_2$ &   $O(1)$ 
   & 
	 Theorem~\ref{thm-simple-up}(\emph{b})
	 & 
   {\Large\ding{174}}
\\
\cmidrule{2-6}
& 
   weighted neighbor
	 & $1+\eps_2+\delta$ &   $O(1)$ 
   & 
	 Corollary~\ref{cor-trans-query}(\emph{i})
	 & 
   {\Large\ding{175}}
\\
\cmidrule{2-6}
& 
   weighted neighbor
	 & $\frac{1}{2}+\eps_2+\delta$ &   $O(1)$ 
   & 
	 Corollary~\ref{cor-trans-query}(\emph{ii})
	 & 
   {\Large\ding{176}}
\\
\bottomrule
\\
\multicolumn{6}{l}{
{\Large\ding{172}} 
\makecell[l]{
\emph{even if} $\deg_{H,1}(x)\leq 1$, $\deg_{H,2}(x)=0$ for every node $x$, and \emph{any} number of weighted 
\\
selective degree queries are allowed.
}
}
\\
[8pt]
\multicolumn{6}{l}{
   {\Large\ding{173}} if 
	 $\deg_{H,1}(x)=O(1)$ for every node $x$,
	 \emph{and} 
	 $\deg_G(u)=\deg_G(v)$. 
}
\\
[8pt]
\multicolumn{6}{l}{
   {\Large\ding{174}} if 
	 \emph{both} $\deg_{H,1}(x)=O(1)$ and $\deg_{H,2}(x)=O(1)$ for every node $x$, and 
	 $\deg_G(u)=\deg_G(v)$. 
}
\\
[8pt]
\multicolumn{6}{l}{
   {\Large\ding{175}} 
	 if $\deg_{H,1}(x)=O(1)$ for every node $x$, and 
      $\deg_G(u)\geq (1 - (\delta/3)\,) \times \deg_G(v)$. 
}
\\
[8pt]
\multicolumn{6}{l}{
   {\Large\ding{176}} 
   if \emph{both} $\deg_{H,1}(x)=O(1)$ and $\deg_{H,2}(x)=O(1)$ for every node $x$, and
   $\deg_G(u)\geq (1 - (\delta/3)\,) \times \deg_G(v)$. 
}
\end{tabular}
}
\vspace*{-0.1in}
\caption{\label{tabb1}A summary of query bounds
for computing $\mathfrak{C}_{G}(e)$;
$\eps_1,\eps_2,\delta$
are arbitrary constants
satisfying 
$0<\eps_1\leq 2$ and 
$\eps_2,\delta>0$.}
\end{table}

\subsection{Lower bounds on number of queries for computing $\mathfrak{C}_{G}(e)$}
\label{suc-sub-lower}

{\bf Note that for query lower bounds it suffices to prove the lower bound for 
complete edge-weighted bipartite graph reformulations of the problem as discussed 
in 
Section~\ref{sec-eqdeg}}.
Any complete bipartite graph $H=(L,R,w)$ used in our lower bound proofs
will satisfy $L\cap R=\emptyset$, thereby implying $\pmb{n=\deg_G(v)-1}$.
Since we provide our inputs in the form of such graphs $H$, we first need to show 
that there exists a graph $G$ with the edge $\{u,v\}$ such that $G_{u,v}=H$ in the notations used
in Section~\ref{sec-eqdeg}. 

\begin{proposition}
Given any complete edge-weighted bipartite graph 
$H=(L,R,w)$ where 
$w:L\times R \mapsto \{1,2,3\}$
there exists a graph $G=(V,E)$ such that $G_{u,v}=H$.
\end{proposition}

\begin{proof}
Start with the edge $\{u,v\}$ in $G$, connect the nodes 
$u_1,\dots,u_n$ to $u$, and 
connect the nodes $v_1,\dots,v_n$ to $v$. 
For every pair of nodes $(u_i,v_j)\in L\times R$, if 
$w(u_j,v_j)=1$ then add the edge $\{u_i,v_j\}$ to $G$. Otherwise if 
$w(u_j,v_j)=2$ then add a new node $x_{i,j}$ to $G$ and add the two edges 
$\{u_i,x_{i,j}\}$ and  
$\{x_{i,j},v_j\}$ 
to $G$.
\end{proof}

A common thread in our lower bound proofs is the following easy but crucial observation.

\begin{observation}\label{obs2}
Suppose that we have two separate classes of 
$($complete edge-weighted bipartite, as described in Section~{\em\ref{sec-eqdeg}}$)$ graphs 
$\G_1$ and $\G_2$,
two numbers $1\leq\alpha<\beta\leq 3$,
and 
an algorithm $\cA$ 
such that the following holds:
\begin{enumerate}[label=$\triangleright$]
\item
Every graph $H\in\G_1$ satisfies 
$n\leq\cM(H)\leq\alpha\,n$.
\item
Every graph $H\in\G_2$ satisfies 
$\beta\,n\leq\cM(H)\leq 3n$.
\item
Given a graph from $\G_1\cup G_2$, 
algorithm $\cA$ cannot determine in which class the given graph belongs.
\end{enumerate}
Then, using Proposition~{\em\ref{obs1}}, it follows that 
algorithm $\cA$ cannot provide an 
additive $(\beta-\alpha-\eps)$-approximation
of $\mathfrak{C}_{G}(e)$
for any constant $\eps>0$.
\end{observation}

Our proofs 
in Theorem~\ref{thm-lower-all}
for lower bounds on the number of queries 
will use the well-known
\emph{Yao's minimax principle} for randomized algorithms~\cite{Yao77}. 
Namely, we will construct two separate classes 
$\G_1$
and 
$\G_2$
of graphs and show that any \emph{deterministic} algorithm that picks graphs uniformly at random from these two classes will need 
\emph{at least} a certain number of queries, say $q$, to be able to decide from which class the graph was selected with at least
a certain probability, say $p$. Then, 
the expected number of queries performed by any deterministic algorithms on inputs drawn from the aforementioned distribution 
is at least $pq$, and thus
by the minimax principle the expected number of queries for any randomized algorithm 
over all possible inputs
is also at least $pq$.
Note that since our input instances are complete bipartite graphs, two graphs are differentiated 
based on the assignments of weights to all possible edges (see~\cite{PARNAS2007} for further elaborations on this point).

\begin{theorem}\label{thm-lower-all}
Consider any local algorithm that is allowed to make an unlimited number of 
weighted selective degree queries.
Let $Q$ be the expected number of queries, 
\textbf{excluding} all weighted selective degree queries,
performed by the algorithm
for computing 
$\mathfrak{C}_{G}(e)$.
Then the following claims hold.
\begin{description}
\item[(\emph{a})]
Suppose that we want to compute
$\mathfrak{C}_{G}(e)$
exactly. 
Then the following bounds hold.
\begin{description}
\item[(\emph{i})]
$Q> \nicefrac{n^2}{6}$ if the queries used are
weighted node-pair queries.
\item[(\emph{ii})]
$Q> \nicefrac{n}{6}$ if the queries used are
weighted neighbor 
queries.
\end{description}
\item[(\emph{b})]
For every $0<\eps<2$,
any randomized algorithm
computing 
an additive 
$(2-\eps)$-approximation
of
$\mathfrak{C}_{G}(e)$
requires 
$Q>\nicefrac{n}{6}$
weighted node-pair queries.
\end{description}
\end{theorem}

\begin{proof}
All the bipartite graphs $H=(L,R,w)$ in our proofs will satisfy that 
$\deg_{H,1}(x)=1$ and $\deg_{H,2}(x)=0$ for every node $x\in L\cup R$, and therefore 
any number of weighted selective degree queries will provide \emph{no} information about the value of 
$\mathfrak{C}_{G}(e)$.

\medskip
\noindent
{\bf Proof of (\emph{a})}

\smallskip
Corresponding to every node pair $(u_i,v_j)$ with $u_i\in L$ and $v_j\in R$, 
the class 
$\G_1$
contains a graph in which $w(u_i,v_j)=1$ and all other edges have weight $3$.
The class 
$\G_2$
contains just one graph in which all edge weights are set to $3$. 
Note that the minimum weight of a perfect matching for each graph in 
$\G_1$
is $3n-2$, whereas
the minimum weight of a perfect matching for the graph in 
$\G_2$
is $3n$.

\medskip
\noindent
\textbf{Proof of (\emph{a})-(\emph{i}) }

\smallskip

Suppose that our algorithm has already made $t$ queries (edges) $e_1,\dots,e_t$
for $t<n^2-1$ with 
$w(e_1)=\dots=w(e_t)=3$
and 
let $e_{t+1}$ be the next query. 
Consider a graph
$G_1\in\G_1$
that is consistent with the first $t$ queries
with $w(e_{t+1})=1$.
Note that there is exactly one such graph in $\G_1$.
Since 
there are at least $n^2-(t+1)$ node pairs (edges)
that have not been queried after the $(t+1)\tx$ query, 
we have at least $n^2-(t+1)$
distinct graphs in $\G_1$ with $w(e_{t+1})=3$
that is consistent with the first $t$ queries
(set the weight of exactly one of the $n^2-(t+1)$ edges to $1$ and the weight of the remaining edges to $3$).
Since 
graphs are selected uniformly at random from 
$\G_1$ it follows that 
$\Pr [ w(e_{t+1})=1 \,|\, w(e_1)=\dots=w(e_t)=3 ]
\leq
\frac{1}{n^2-(t+1)}
$.
Summing over all $t$, we get 
\begin{multline*}
\Pr [\text{number of queries needed is at least $t+1$}]  
\\
=
1 - \Pr [\text{one of the $t$ queries contain an edge of weight $1$}]
\\
\geq 
1-{\textstyle \frac{t}{n^2-(t+1)} }
\end{multline*}
Putting $t=\nicefrac{n^2}{3}$, the probability that ``the number of queries is at least $1+\nicefrac{n^2}{3}$'' 
is at least $\nicefrac{1}{2}$.

\medskip
\noindent
\textbf{Proof of (\emph{a})-(\emph{ii}) }
\smallskip

Suppose that our algorithm has already made $t$ queries (nodes, weights) $(x_1,y_1),\dots,(x_t,y_t)\in ( L\cup R ) \times \{1,2,3\}$
for $t<n-1$. 
Let 
$e_1=\{x_1,x_1'\},\dots,e_t=\{x_t,x_t'\}$
be the answers (edges) to these queries 
with 
$w(e_1)=\dots=w(e_t)=3$
and 
let 
$(x_{t+1},y_{t+1})$
be the next query that reveals the weight of an edge $e_{t+1}=\{x_{t+1},x_{t+1}'\}$.
Consider a graph
$G_1\in\G_1$
that is consistent with the first $t$ queries
with $w(e_{t+1})=y_{t+1}=1$.
Note that there is exactly one such graph in $\G_1$.
Since there are at least $n-(t+1)$ nodes
in each of $L$ and $R$ 
that have not been queried after the $(t+1)\tx$ query, 
we have at least $(n-(t+1))^2$
distinct graphs in $\G_1$ with $w(e_{t+1})=3$
that is consistent with the first $t$ queries
(set the weight of exactly one edge among these nodes to $1$ and 
the weights of all remaining edges to $3$). 
Since 
graphs are selected uniformly at random from 
$\G_1$ it follows that 
$\Pr [ w(e_{t+1})=1 \,|\, w(e_1)=\dots=w(e_t)=3 ]
\leq
\frac{1}{(n-(t+1))^2}
$.
Summing over all $t$, we get 
\begin{multline*}
\Pr [\text{number of queries needed is at least $t+1$}]  
\\
=
1 - \Pr [\text{one of the $t$ queries contain an edge of weight $1$}]
\\
\geq 
1-{\textstyle \frac{t}{(n-(t+1))^2} }
\end{multline*}
Putting $t=\nicefrac{n}{2}$, the probability that ``the number of queries is at least $1+\nicefrac{n}{2}$'' 
is at least $1-\frac{2}{n}$.

\bigskip
\noindent
\textbf{Proof of (\emph{b})}
\smallskip

Corresponding to each of the possible $n!$ perfect matchings, 
the class 
$\G_1$
contains a graph in which the edges in the matching have weight $1$ and all other non-matching edges have weight $3$.
The class 
$\G_2$
contains just one graph in which all edge weights are set to $3$. 
Note that the minimum weight of a perfect matching for each graph in 
$\G_1$
is $n$, whereas
the minimum weight of a perfect matching for the graph in 
$\G_2$
is $3n$. 
Suppose that our algorithm has already made $t$ (edge) queries $e_1,\dots,e_t)$
for $t<n-1$ with 
$w(e_1)=\dots=w(e_t)=3$
and 
let $e_{t+1}$ be the next (edge) query.

We first show that as long as $t<n$ there exists at least one graph in 
$\G_1$
that is consistent with the weight assignments of the first $t$ queries. 
Consider a random perfect matching 
$M= \{ \, \{u_1,v_{\pi(1)} \} ,\dots,  \{u_n,v_{\pi(n)} \} \, \}$
given by a random permutation $\pi$ of $1,\dots,n$.
The probability of the event $\cE_j$ that the $j\tx$ query $e_j$ is in $M$
is $\frac{(n-1)!}{n!}=\nicefrac{1}{n}$. It follows 
that
$
\Pr [ \wedge_{j=1}^t \overline{\cE_j} ] 
=
1 - 
\Pr [ \vee_{j=1}^t \cE_j ] 
\geq 
1 - 
\sum_{j=1}^t \Pr [ \cE_j ] 
\geq 1 -\frac{t}{n}
>0
$
and therefore 
$\G_1$
contains at least one such graph.

Assume 
without loss of generality 
that
$e_{t+1}=(u_n,v_n)$
and let $M$ be 
a perfect matching of the nodes in $L$ and $R$, 
say 
$M= \{ \, \{u_1,v_1 \} ,\dots,  \{u_n,v_n \} \, \}$, 
that is consistent with the first $t$ queries,
and includes $e_{t+1}$ as a matched edges (note that 
$w(u_1,v_1)=\dots= w(u_n,v_n)=1$).
If such a matching $M$ does not exist then 
$\Pr [ w(u_{n},v_{n})=1 \,|\, w(e_1)=\dots=w(e_t)=3 ] = 0$.
Otherwise, 
note that there are at least $n-t$ nodes in each of $L$ and $R$, say $u_1,\dots,u_{n-t}\in L$ and $v_1,\dots,v_{n-t}\in R$,
such that the edges $(u_{n},v_j)$ and $(u_j,v_{n})$ for $j=1,\dots,n-t$ have not been queried yet.
For every such perfect matching $M$, we can then construct a set 
$S_M$ of 
at least $n-t$
\emph{distinct} perfect matchings with 
$w(e_{t+1})=3$
that is consistent with the first $t$ queries
as follows: in the 
$\ell\tx$
perfect matching  
set
$w(u_\ell,v_\ell) = w(u_{n},v_{n}) =3$
and 
set 
$w(u_{n},v_\ell) = w(u_\ell,v_{n}) =1$.
It is also easy to see that any two matchings from two different sets $S_M$ and $S_{M'}$ 
differ in at least one edge.
Since 
graphs are selected uniformly at random from 
$\G_1$ it follows that 
$\Pr [ w(u_{n},v_{n})=1 \,|\, w(e_1)=\dots=w(e_t)=3 ]
\leq
\frac{1}{n-t}
$.
Summing over all $t$, we get 
\begin{multline*}
\Pr [\text{number of queries needed is at least $t+1$}]  
\\
=
1- \Pr [\text{any of the $t$ queries contain an edge of weight $1$}]
\\
> 
1- {\textstyle \frac{t}{n-(t-1)} }
\end{multline*}
Putting $t=\nicefrac{n}{3}$, the probability that ``the number of queries is at least $1+\nicefrac{n}{3}$'' 
is at least $\nicefrac{1}{2}$.
\end{proof}

\subsection{Upper bounds on number of queries for computing $\mathfrak{C}_{G}(e)$ when $\deg(u)=\deg(v)$}
\label{suc-sub-upper}

The proofs 
in Theorem~\ref{thm-lower-all}
do \emph{not} use any edge of weight $2$ and have \emph{at most} one edge of weight $1$ 
incident on any node with the additional restriction that these edges of weight $1$ provide 
a unique matching for the nodes that are end-points of these edges. In this section we show that if 
weighted neighbor queries are allowed then $O(1)$ expected number of queries will suffice for a non-trivial
additive approximation for a class of weighted complete bipartite graphs that properly includes 
the instances generated by the proofs 
in Theorem~\ref{thm-lower-all} (note that 
for the instances (graphs) generated by the proofs 
in Theorem~\ref{thm-lower-all}
we have $\deg_{H,1}(x)\leq 1$ and $\deg_{H,2}(x)=0$ for \emph{every} node $x$).

\begin{theorem}\label{thm-simple-up}
Assume that 
$\deg_G(u)=\deg_G(v)$, and 
let $d,\eps>0$ be two fixed constants. 
Then, using $O(1)$ 
expected number of 
weighted neighbor queries\footnote{The constant in $O(1)$ depends on $d$ and $\eps$.} 
we can obtain the following type of approximations for 
$\mathfrak{C}_{G}(e)$:
\begin{description}
\item[(\emph{a})]
an additive $\left(1+\eps\right)$-approximation when 
$\max_{x} \{ \deg_{H,1}(x) \} \leq d$, 
and 
\item[(\emph{b})]
an additive $\left(\frac{1}{2}+\eps\right)$-approximation when 
$\max_{x} \{ \deg_{H,1}(x) \} \leq d$ and 
\\
$\max_{x} \{ \deg_{H,2}(x) \} \leq d$.
\end{description}
\end{theorem}

\begin{remark}
Let $m_1$, $m_2$ and $m_{12}$ 
be as defined in the proof of this theorem.
The bounds in Theorem~{\em\ref{thm-simple-up}}
are tight in the sense that no algorithm that knows only estimates of $m_1$ $($resp.\ estimates of $m_1,m_2,m_{12})$ 
can provide better additive ratios for parts 
$($a$)$ $($resp.\ $($b$))$; 
see {\em\FI{fig1}(\emph{a})--(\emph{b})}.
The example in 
{\em\FI{fig1}(\emph{c})}
shows that no algorithm can provide better than additive $\frac{2}{3}$-approximation 
for the case in 
Theorem~{\em\ref{thm-simple-up}(\emph{b})}
if the estimate for $m_{12}$ is not used.
\end{remark}

\begin{figure}[htbp]
\includegraphics[scale=0.83]{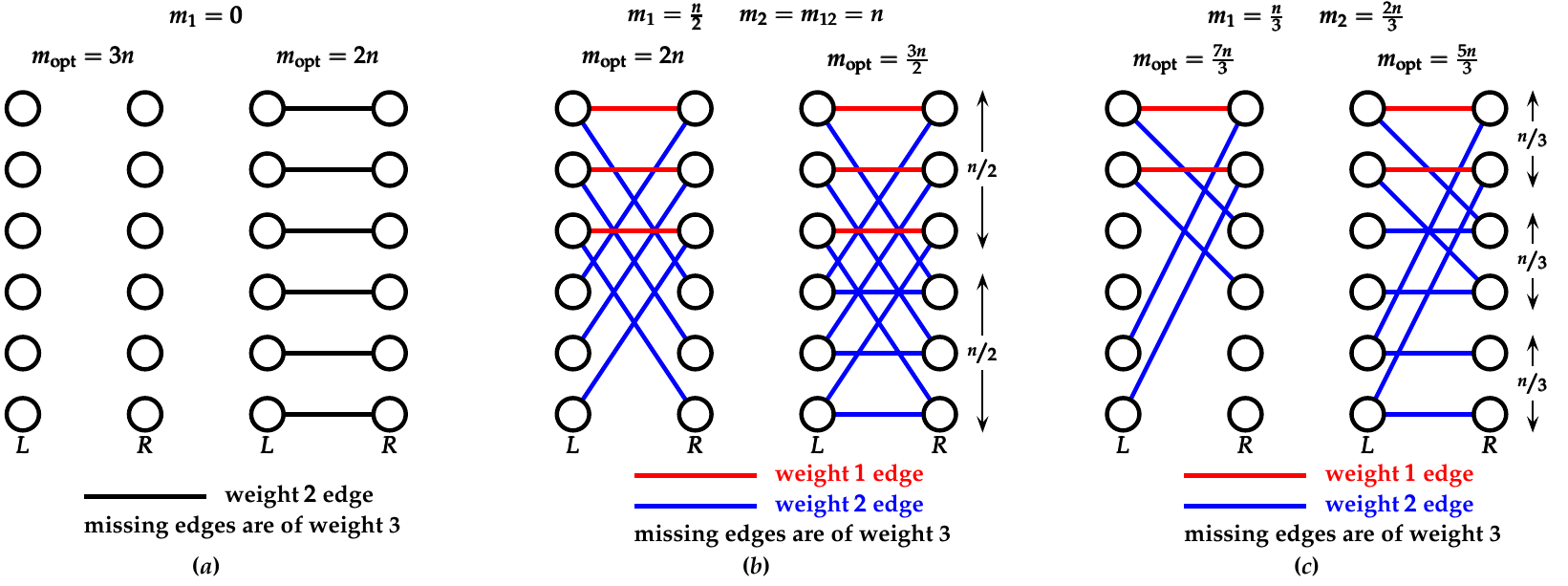}
\caption{\label{fig1}
\textbf{(\emph{a})}
Example showing tightness of bounds in Theorem~\ref{thm-simple-up}
when only estimate for $m_1$ is known.
\textbf{(\emph{b})}
Example showing tightness of bounds in Theorem~\ref{thm-simple-up}
when estimates for $m_1,m_2,m_{12}$ are known;
\textbf{(\emph{c})}
Example showing that
better than additive $\frac{2}{3}$-approximation 
is not possible 
if only 
estimates for $m_1$ and $m_2$ are used
for the case in 
Theorem~\ref{thm-simple-up}(\emph{b}).
}
\end{figure}

\begin{proof}
Since 
$\deg_G(u)=\deg_G(v)$ 
we can use the reformulations of the linear program~\eqref{lpcg} 
outlined in Section~\ref{sec-eqdeg}.
Let $\delta>0$ be a \emph{constant} to be fixed later.
Let 
$H_1$, $H_2$ and $H_{12}$ 
be the subgraphs of $H$ induced by the edges in $H$ of weight $1$, edges in $H$ of weight $2$, and edges in $H$ of weights 
$1$ and $2$, respectively.
Fix \emph{maximum-cardinality} matchings 
$\cM_1$, $\cM_2$ and  $\cM_{12}$ 
of 
$H_1$, $H_2$ and $H_{12}$ 
having 
$m_1\,n$, $m_2\,n$ and $m_{12}\,n$ 
edges, respectively.
Also, 
fix a minimum-weight perfect matching $\Mopt$ of $H$ of \emph{total} weight 
$\mopt\,n$, 
and let 
$
m_{\mathrm{opt},\ell} \,n
$ 
be the number of edges of weight $\ell\in\{1,2,3\}$ in $\Mopt$.
The following inequalities will be useful during the rest of the proof:
\begin{gather*}
m_{\mathrm{opt},1}\leq m_1, 
\,\,\,\,
m_{\mathrm{opt},2}\leq m_2,
\,\,\,\,
m_{\mathrm{opt},3}\geq 1-m_{12},
\,\,\,\,
m_{12} \geq \max\{m_1,m_2\},
\end{gather*}
\begin{multline*}
\mopt=
m_{\mathrm{opt},1} +2\,m_{\mathrm{opt},2} + 3(1-m_{\mathrm{opt},1} - m_{\mathrm{opt},2} )=
3-2\,m_{\mathrm{opt},1}-m_{\mathrm{opt},2}
\\
\geq 2-2\,m_{\mathrm{opt},1}
\geq
2 - 2\,m_1
\end{multline*}
Let 
$\M_s$ be a perfect matching of $H$ generated by 
taking all the edges (of weight $1$) in 
$\cM_1$ and pairing the remaining nodes from $L$ and $R$ arbitrarily.
Note that the total weight $m_s \,n$ of the edges in $\M_s$ satisfies 
$\mopt \leq m_s$ and 
$m_s \leq m_1  + 3(1-m_1)=3-2\,m_1$; thus it follows that 
$3-2\,m_1\geq\mopt$.
Similarly, 
taking 
$\M_s$ to be a perfect matching of $H$ of total weight $m_s\,n$ generated by 
taking all the edges (of weight $2$) in 
$\cM_2$ and pairing the remaining nodes from $L$ and $R$ arbitrarily 
we get 
$\mopt \leq m_s$ and 
$m_s \leq 2\,m_2  + 3(1-m_2)=3-m_2$; thus it follows that 
$3-m_2\geq\mopt$.

Our algorithm proceeds in two main steps. The first common step in our algorithm for both 
\textbf{(\emph{a})} and 
\textbf{(\emph{b})}
is to determine the set of nodes 
in $L$ and $R$ from $\nbr_G(u)$ and $\nbr_G(v)$.
This can be done by comparing the list of nodes in 
$\nbr_G(u)$ and $\nbr_G(v)$ to identify all nodes in 
$\nbr_G(u)\cap \nbr_G(v)$ 
and setting 
$L= \nbr_G(u) \setminus \big( \nbr_G(u)\cap \nbr_G(v) \big), R= \nbr_G(v) \setminus \big( \nbr_G(u)\cap \nbr_G(v) \big)$. 
Note that this step does \emph{not} use any query at all.
The remaining parts of our algorithms will only use 
weighted neighbor queries $(x,s)$ for $x\in L\cup R$ and $s\in\{1,2\}$.

\medskip
\noindent
\textbf{Proof of (\emph{a})}

\smallskip
Since 
$\max_{v\in L\cup R} \{ \deg_{H_1}(v) \} \leq d$, 
then 
using the results of 
Yoshida \EA~\cite{YYI2012}
we can 
compute a number 
$\widetilde{m}_1$
using 
$d^{O(1/\delta^2)} {( \nicefrac{1}{\delta} )}^{O(1/\delta)}  =O(1)$
expected number of 
queries
such that 
$m_1\,n-\delta\,n \leq \widetilde{m}_1\,n \leq m_1\,n$.
It is straightforward to see that each query in 
Yoshida \EA~\cite{YYI2012}
can be implemented by a 
weighted neighbor query $(x,1)$ for some appropriate $x\in L\cup R$.
After using $O(1)$ 
expected number of 
weighted neighbor queries to compute 
$\widetilde{m}_1\,n$
we output 
the number 
$
\Delta=
(3- 2\,\widetilde{m}_1)
$
as our estimate for $\mopt$.
Note that 
$\Delta \geq (3- 2\,{m}_1) \geq \mopt $, 
and 
$
\Delta - \mopt
=
(3- 2\,\widetilde{m}_1) - \mopt
\leq 
(3- 2\,{m}_1) + 2\,\delta - \mopt
\leq 
1 + 2\,\delta
$.
Our proof is completed by taking $\delta=\nicefrac{\eps}{2}$.

\medskip
\noindent
\textbf{Proof of (\emph{b})}

\smallskip
Since  
$\max_{v\in L\cup R} \{ \deg_{H_1}(v) \} \leq d$
and 
$\max_{v\in L\cup R} \{ \deg_{H_2}(v) \} \leq d$, 
using the results of 
Yoshida \EA~\cite{YYI2012}
we can compute numbers 
$\widetilde{m}_1$, $\widetilde{m}_2$, and  $\widetilde{m}_{12}$ 
using 
$(2d)^{O(1/\delta^2)} {( \nicefrac{1}{\delta} )}^{O(1/\delta)}  =O(1)$
expected number of 
queries 
such that 
$m_\ell\,n-\delta\,n \leq \widetilde{m}_\ell\,n \leq m_\ell\,n$ for $\ell\in\{1,2,12\}$.
It is straightforward to see that each query in 
Yoshida \EA~\cite{YYI2012}
can be implemented by a 
weighted neighbor query $(x,s)$ for some appropriate $x\in L\cup R$ and $s\in\{1,2\}$.
We perform the following case analysis to provide our estimate $\Delta$ of $\mopt$.
\begin{description}[leftmargin=0.2in]
\item[Case~1: $\widetilde{m}_1\leq\nicefrac{1}{4}$.]
Our estimate for $\mopt$ is
$\Delta=
3 - \widetilde{m}_2
$.
Note that 
$
\Delta
\geq
3 - m_2\geq\mopt
$.
For the additive error estimation, 
we have 
\begin{multline*}
\Delta - \mopt = 
3 - \widetilde{m}_2 - \mopt
\leq 
(3 - {m}_2 -2m_1) + 2m_1 +\delta - \mopt
\\
\leq 
(3-m_{\mathrm{opt},2}-2\,m_{\mathrm{opt},1}) + 2 \big( {\textstyle \frac{1}{4} } +\delta\big)  + \delta - \mopt
=
{\textstyle \frac{1}{2} } + 3\,\delta
\end{multline*}
\item[Case~2: $\widetilde{m}_2\leq\nicefrac{1}{2}$ or $\widetilde{m}_1\geq\nicefrac{1}{2}$ or $\widetilde{m}_{12}\leq\nicefrac{3}{4}$.]
Our estimate for $\mopt$ is
$\Delta=
3 - 2\,\widetilde{m}_1
$.
Note that 
$
\Delta
\geq
3 - 2\,m_1
\geq\mopt
$.
For the additive error bounds, 
we have the following:
\begin{enumerate}[label=$\triangleright$]
\item
If $\widetilde{m}_2\leq\nicefrac{1}{2}$ then 
$
\Delta - \mopt
=
3-2 \, \widetilde{m}_1 - \mopt
=
(3-m_2-2\,m_1) + m_2  + 2\,\delta - \mopt
\leq
(3-m_{\mathrm{opt},2}-2\,m_{\mathrm{opt},1}) + \nicefrac{1}{2} + 3\,\delta - \mopt
\leq
{\textstyle \frac{1}{2} } + 3\,\delta
$.
\item
If $\widetilde{m}_1\geq\nicefrac{1}{2}$ then 
since the \emph{smallest possible} total weight that any perfect matching of $H$ could have is achieved by taking 
all the $m_1\,n$ edges of weight $1$ and 
the remaining $(1-m_1)\,n$ edges of weight $2$
we get $\mopt\geq m_1 + 2(1-m_1)=2-m_1$.
Consequently, 
\[
\Delta- \mopt
\leq 
(3 - 2\,\widetilde{m}_1) - 
(2-m_1)
\leq
1-m_1 + 2\,\delta
\leq
{\textstyle \frac{1}{2} } + 2\,\delta
\]
\item
If $\widetilde{m}_{12}\leq\nicefrac{3}{4}$ then 
since $m_{\mathrm{opt},3}\geq 1-m_{12}$
the \emph{smallest possible} total weight that any perfect matching of $H$ could achieve is by taking 
$m_1\,n$ edges of weight $1$,
$(1-m_{12})n$ edges of weight $3$,
and 
the remaining 
$(m_{12}-m_1)n$
edges of weight $2$
we get $\mopt\geq m_1 + 2(m_{12}-m_1) +3(1-m_{12})=3-m_1-m_{12}$.
Consequently,
\begin{multline*}
\Delta -\mopt
\leq 
( 3 - 2\,\widetilde{m}_1 ) - ( 3-m_1-m_{12} )
\leq
(m_{12}-m_1) + 2\,\delta 
\\
\leq
( {\textstyle \frac{3}{4} } + \delta - {\textstyle \frac{1}{4} }  ) + 2\,\delta 
=
{\textstyle \frac{1}{2} } + 3\,\delta
\end{multline*}
\end{enumerate}
\item[Case~3: when Cases $1$ and Case $2$ do not apply.]
For this case the following inequalities hold: 
\begin{gather*}
\nicefrac{1}{4} < \widetilde{m}_1 < \nicefrac{1}{2} 
\,\Rightarrow\,
\nicefrac{1}{4} < {m}_1 < \nicefrac{1}{2} + \delta,
\,\,\,
\widetilde{m}_2 > \nicefrac{1}{2}
\,\Rightarrow\,
{m}_2 > \nicefrac{1}{2} + \delta
\\
\widetilde{m}_{12} > \nicefrac{3}{4}
\,\Rightarrow\,
{m}_{12} > \nicefrac{3}{4}
\end{gather*}
For this case, we use the following lower bound for $\mopt$.
Since $m_{\mathrm{opt},3}\geq 1-m_{12}$
the \emph{smallest possible} total weight that any perfect matching of $H$ could have is achieved by taking 
$m_1\,n$ edges of weight $1$,
$(1-m_{12})n$ edges of weight $3$,
and 
the remaining 
$(m_{12}-m_1)n$
edges of weight $2$.
This implies $\mopt\geq m_1 + 2(m_{12}-m_1) +3(1-m_{12})=3-m_1-m_{12}$.

Let $\alpha= \max \{ {m}_{12}-2\,{m}_1,\, 0 \}$.
Suppose that $\cM_{12}$ contains $m_1'\leq m_1$ edges of weight $1$.
Consider the following process: we start with the edges in $\cM_{12}$, remove $m_1'$ edges 
of weight $1$ from it, add $m_1$ edges of weight $1$ from $\cM_1$ to it and finally remove (``knock out'')
the edges of weight $2$ that share an end-point with the edges of $\cM_1$ added to our collection.
Since $m_1$ edges of weight $1$ can knock out \emph{at most} $2\,m_1$ edges of weight $2$, 
it follows that there are at least $\alpha$ ``surviving'' edges of weight $2$ that do \emph{not} share any end-point 
with the edges in $\cM_1$.
We now have the following two sub-cases.
\begin{description}[leftmargin=0.2in]
\item[Case~3.1: $\widetilde{m}_{1,2}\leq 2\,\widetilde{m}_1+\delta$.]
Note that 
$\widetilde{m}_{1,2}\leq 2\,\widetilde{m}_1+\delta$
implies 
${m}_{1,2}\leq 2\,{m}_1 + 2\,\delta$.
Our estimate for $\mopt$ is
$\Delta=
3 - 2\,\widetilde{m}_1
$.
Note that 
$
\Delta
\geq
3 - 2\,m_1
\geq\mopt
$.
For the additive error estimation, 
note that
\begin{multline*}
\Delta - \mopt
\leq 
( 3 - 2\,\widetilde{m}_1 ) - (3-m_1-m_{12})
\leq 
m_{12}-m_1 + 2\,\delta
\\
\leq 
m_1 + 4\,\delta
<
{\textstyle \frac{1}{2} } + 5\,\delta
\end{multline*}
\item[Case~3.2: $\widetilde{m}_{1,2}>2\,\widetilde{m}_1+\delta$.]
Our estimate for $\mopt$ is
$\Delta= 3 - \widetilde{m}_{1,2}$.
Note that 
$\widetilde{m}_{1,2}> 2\,\widetilde{m}_1+\delta$
implies 
${m}_{1,2}> 2\,{m}_1 + \delta$.
Thus, for this case, $\alpha=m_{12}-2\,m_1>\delta>0$.
Let $\M'$ be a perfect matching of $H$ generated by 
taking all the edges (of weight $1$) in 
$\cM_1$, the $\alpha$ surviving edges of weight $2$, and pairing the remaining nodes from $L$ and $R$ arbitrarily.
Then, the total weight $m' \,n$ of the edges in $\M'$ satisfies 
$
m_1  + 2 (m_{12}-2\,m_1) + 3(1-(m_1 + (m_{12}-2\,m_1) ) ) = 3 -m_{12}
\geq 
m' 
\geq 
\mopt
$,
and
it follows that 
$
\Delta
\geq
3 - m_{1,2}
\geq\mopt
$.
For the additive error estimation, 
note that
\begin{gather*}
\Delta - \mopt 
\leq 
( 3 - \widetilde{m}_{1,2} ) - ( 3-m_1-m_{12} )
\leq 
m_1 + \delta
<
{\textstyle \frac{1}{2} } + 2\,\delta
\end{gather*}
\end{description}
\end{description}
In all cases, setting $\delta=\nicefrac{\eps}{5}$ provides an additive $\left(\frac{1}{2}+\eps\right)$-approximation.
\end{proof}

\subsection{Upper bounds on number of queries for computing $\mathfrak{C}_{G}(e)$ when $\deg(u)\neq\deg(v)$
using ``localized'' fine-grained reduction}
\label{suc-sub-upper2}

\newcommand{\bbbeq}{\mathfrak{B}_=}
\newcommand{\bbbless}{\mathfrak{B}_<}

Theorem~\ref{thm-simple-up}
provides non-trivial approximation of 
$\mathfrak{C}_{G}(e)$ 
when $\deg_G(u)=\deg_G(v)$. 
In this section, we show that
a ``localized'' version of the fine-grained reduction used in 
Theorem~\ref{thm-ineqdeg}
can be applied to extend these local approximation algorithms to some cases when 
$\deg_G(u)$ and $\deg_G(v)$ are \emph{not} necessarily equal. 
Such 
a localized version of the fine-grained reduction is \emph{not} allowed to construct the reduction explicitly, 
but instead the details of the reduction need to be revealed 
incrementally to the local algorithm on a ``need-to-know'' basis to simulate the queries of the local algorithm 
on the graph constructed by the fine-grained reduction.
The overall simulation is summarized in Theorem~\ref{thm-trans-query}.

\begin{theorem}[\bf Computing $\pmb{\mathfrak{C}_{G}(e)}$ via localized fine-grained reduction]\label{thm-trans-query}
Suppose that we have an algorithm 
$\bbbeq$ that provides an 
$(\alpha,\eps)$-estimate for 
$\mathfrak{C}_{G}(e)$ 
when $\deg_G(u)=\deg_G(v)$ using $t$ queries 
$q_1',\dots,q_t'$ 
when each query $q_i'$ 
is either a weighted node-pair query, a weighted neighbor query or 
a weighted selective degree query.

Then, letting $\delta>0$ denote any constant, we can 
design an algorithm 
$\bbbless$ 
for the case when 
$\deg_G(u)\neq\deg_G(v)$
using 
$\bbbeq$
with the following properties:
\begin{enumerate}[label={\emph{(}{\alph*}\emph{)}}]
\item
Corresponding to each query $q_i'$,
$\bbbless$ 
performs at most 
one weighted selective degree query
and at most 
one additional query of the same type as $q_i'$
on $G_{u,v}$. 
\item
$\bbbless$ 
provides an 
$(\alpha,\eps)$-estimate for 
$\mathfrak{C}_{G}(e)$ 
if 
$\deg_G(u)+1$ is an integral multiple of $\deg_G(v)+1$. 
\item
$\bbbless$ 
provides an 
$(\alpha,\eps+\delta)$-estimate for 
$\mathfrak{C}_{G}(e)$
if 
either 
$\deg_G(u)\leq(\delta/3)\times \deg_G(v)$
or 
$\deg_G(u)\geq (1 - (\delta/3)\,) \times \deg_G(v)$. 
\end{enumerate}
\end{theorem}

\begin{corollary}\label{cor-trans-query}
If $\deg_G(u)\geq (1 - (\delta/3)\,) \times \deg_G(v)$ for some constant $\delta>0$ 
then 
$\max_{x} \{ \deg_{G_{u,v},1}(x) \}=O(1)$
$($respectively, $\max_{x} \{ \deg_{G_{u,v},2}(x) \}=O(1))$ 
implies 
$\max_{x} \{ \deg_{G_{u,v}',1}(x) \}=O(1)$
$($respectively, $\max_{x} \{ \deg_{G_{u,v}',2}(x) \}=O(1))$, and thus 
each weighted selective degree query
for the weight $1$ $($respectively, for the weight $2)$
can be trivially simulated by $O(1)$ 
weighted neighbor queries 
for the weight $1$ $($respectively, for the weight $2)$
on $G_{u,v}'$.
Thus, 
combining Theorem~\ref{thm-trans-query} with the approximations in Theorem~\ref{thm-simple-up}
gives us algorithms of the following types for the case when $\deg_G(u)\neq\deg_G(v)$:
\begin{enumerate}[label={\emph{(}{\roman*}\emph{)}}]
\item
additive $\left(1+\eps+\delta\right)$-approximation 
using $O(1)$ 
weighted neighbor 
queries\footnote{The constant in $O(1)$ depends on the value of $\frac{1}{1-(\delta/3)}$.}
if 
$\max_{x} \{ \deg_{H,1}(x) \}=O(1)$ and 
$\deg_G(u)\geq (1 - (\delta/3)\,) \times \deg_G(v)$,
\item
additive $\left(\frac{1}{2}+\eps+\delta\right)$-approximation 
using $O(1)$ 
weighted neighbor 
queries\footnotemark[\value{footnote}]
if 
$\max_{x} \{ \deg_{H,1}(x) \}=O(1)$, 
$\max_{x} \{ \deg_{H,2}(x) \}=O(1)$, and 
$\deg_G(u)\geq (1 - (\delta/3)\,) \times \deg_G(v)$.
\end{enumerate}
\end{corollary}

\begin{proof}
We will \emph{reuse} the notations and the reduction used in the proof of
Theorem~\ref{thm-ineqdeg}; in particular in those notations the graph $H$ is also the graph $G_{u,v}$.
Our algorithm 
$\bbbless$ has a list of nodes in the graph $G_{u,v}'$ and also the numbers $a$ and $b$.
We show next how the value of a query $q_i'$ on $G_{u,v}'$ can be obtained from
the values of a collection $\cQ_i$ of (at most two) queries 
on $G_{u,v}$ by 
$\bbbless$.
\begin{enumerate}[label=$\triangleright$,leftmargin=*]
\item
\textbf{Case 1: 
$\pmb{q_i'}$ is a 
weighted node-pair query}.
If $q_i'$ is of the form 
$(x_i^j,y_\ell)$ 
then $\cQ_i=\{ (x_i,y_\ell) \}$ and 
$\bbbless$
returns the value of the query 
$(x_i,y_\ell)$ 
on $G_{u,v}$
as the value of $q_i'$. 
If
$q_i'$ is of the form $(r_i,y_\ell)$ then 
$\cQ_i=\emptyset$ and 
$\bbbless$
returns $3$
as the value of $q_i'$. 
\item
\textbf{Case 2: 
$\pmb{q_i'}$ is a 
weighted selective degree query}.
Let $s$ be a number from the set $\{1,2,3\}$.
\begin{itemize}
\item
If $q_i'$ is of the form 
$(x_i^j,s)$
then $\cQ_i=\{ (x_i,s) \}$ 
and 
$\bbbless$
returns the value of the
weighted selective degree query $(x_i,s)$
on $G_{u,v}$
as the value of $q_i'$. 
\item
If $q_i'$ is of the form 
$(r_i,s)$
then $\cQ_i=\emptyset$ and 
$\bbbless$
returns 
$3\deg_G(v)+3$ if $s=3$
and 
$0$
otherwise
as the value of $q_i'$. 
\item
If $q_i'$ is of the form 
$(y_\ell,s)$ 
then $\cQ_i=\{ (y_\ell,s) \}$, and 
$\bbbless$
returns the following 
as the value of $q_i'$:
\begin{itemize}
\item
the value of the weighted selective degree query 
$(y_\ell,s)$ 
on $G_{u,v}$
times $a$
if $s\in\{1,2\}$, and 
\item
the value of the weighted selective degree query 
$(y_\ell,s)$ 
on $G_{u,v}$
times $a$
plus $b$
otherwise.
\end{itemize}
\end{itemize}
\item
\textbf{Case 3: 
$\pmb{q_i'}$ is a 
weighted neighbor query}.
Let $s$ be a number from the set $\{1,2,3\}$.
\begin{enumerate}[label=$\triangleright$,leftmargin=*]
\item
\textbf{Case 3.1: 
$\pmb{q_i'}$ is of the form 
$\pmb{(x_i^j,s)}$}. 
The following example illustrates the subtlety of this case. Suppose that $x_1$ 
is connected to four nodes $y_1,y_2,y_3,y_4$ via edges of weight $s$ in $G_{u,v}$. 
Then each of the nodes $x_1^1,\dots,x_1^a$ is connected to $y_1,y_2,y_3,y_4$ via 
edges of weight $s$ in $G_{u,v}'$. 
\begin{itemize}
\item
As a first attempt, one may simulate the answer to the query 
$(x_1^j,s)$
by performing a query 
$(x_1,s)$ on 
$G_{u,v}$. However, this will not provide new nodes with the correct probabilities required for
random uniform selection among not-yet-explored nodes. For example, suppose that 
$\bbbless$
already made the query 
$(x_1^1,s)$ giving the node $y_1$. 
If now 
$\bbbless$
makes another query 
$(x_1^2,s)$ then such a simulation will return 
a node uniformly randomly from the set of nodes 
$\{y_2,y_3,y_4\}$ but the correct simulation would have been to select 
a node uniformly randomly from the set of nodes 
$\{y_1,y_2,y_3,y_4\}$.
Moreover, if 
$\bbbless$
has already made the queries 
$(x_1^1,s), (x_1^2,s), (x_1^3,s), (x_1^4,s)$
using such a simulation then this simulation of a new query 
$(x_1^5,s)$
will simply return the special symbol.
\item
As a second attempt, to simulate the answer to a query 
$(x_1^j,s)$
one may first check if the answer to a query 
$(x_1^{j'},s)$
for some $j'\neq j$ is already available, and if so simply return that answer. 
But, in this case, the answers to the queries 
$(x_1^j,s)$
and 
$(x_1^{j'},s)$
will \emph{not} be statistically independent.
\end{itemize}
To address 
these and other 
subtleties
we design 
Algorithm 
$\bbbless$
to handle all queries of the form 
$(x_i^j,s)$
\textbf{for each specific} $i$ and $s$
in the following manner. 
Let $\cS_{x_i,s}$ be the set of (not initially known to 
$\bbbless$) 
$\sigma_{i,s}=|\cS_{x_i,s}| $ nodes connected to $x_i$ in $G_{u,v}$ via edges of weight $s$.
\begin{description}
\item[(\emph{i})]
If \emph{not} already done before, we make one new weighted selective degree query 
$(x_i,s)$ 
on $G_{u,v}$ 
giving us the value of 
$\sigma_{i,s}$ (if the value of $\sigma_{i,s}$ \emph{is} already available we simply use it without making a query).
\item[(\emph{ii})]
For each $x_i^j$, we keep a count 
$\kappa_{x_i^j,s}$
of how many times the query 
$(x_i^j,s)$
has been asked involving the node 
$x_i^j$ 
before the current query 
and store the answers to these queries in a set 
$\cT_{x_i^j,s}$.
We also maintain $\cT_{i,s}= \cup_{j=1}^a \cT_{x_i^j,s}$
and 
$\kappa_{i,s}= |\cT_{i,s}|$.
Note the following:
\begin{itemize}
\item
If 
$
{
\kappa_{x_i^j,s} < \sigma_{i,s} 
}
$
then 
performing a new weighted neighbor query 
$\pmb{(x_i^j,s)}$ on $\pmb{G_{u,v}'}$
\emph{must} return
a node uniformly at random from 
the set of nodes 
$
{\Lambda_{x_i^j,s} = 
\cS_{x_i,s}
\setminus 
\cT_{x_i^j,s}
}$
with probability
${\nicefrac{1}{ \lambda_{x_i^j,s} }}$
where
$\lambda_{x_i^j,s} = |\Lambda_{x_i^j,s}|=
\sigma_{i,s} - \kappa_{x_i^j,s}
$.
\item
If
$
{
\kappa_{i,s} < \sigma_{i,s} 
}
$
then
performing a new weighted neighbor query 
${(x_i,s)}$ on $\pmb{G_{u,v}}$
returns
a node uniformly at random from 
the set of nodes 
$
{
\Lambda_{i,s} = 
\cS_{x_i,s}\setminus \cT_{i,s}
}
$
with probability
$
\nicefrac {1} {
\lambda_{i,s}
}
$
where 
$\lambda_{i,s} = | \Lambda_{i,s} | = 
\sigma_{i,s} 
-
\kappa_{i,s}  
$.
\item
Note that we know all the elements of 
$\cT_{i,s}$;  
in particular, this means that \textbf{we \emph{can} sample a node from a subset of 
$\pmb{\cT_{i,s}}$ uniformly at random}.
\end{itemize}
\item[(\emph{iii})]
For a query 
$(x_i^j,s)$, 
we have the following cases.
\begin{enumerate}[label=$\blacktriangleright$,leftmargin=*]
\item
\textbf{Case I: $\pmb{\kappa_{i,s}= \sigma_{i,s}}$}. 
In this case 
$\cT_{i,s}= \cS_{x_i,s}$.
\begin{enumerate}[label=$\blacktriangleright$,leftmargin=*]
\item
\textbf{Case I-a: $\pmb{\kappa_{x_i^j,s}<\kappa_{i,s}}$}. 
We select a node uniformly at random from the set 
$
\cT_{i,s} \setminus \cT_{x_i^j,s}
=
\cS_{x_i,s} \setminus \cT_{x_i^j,s}
$
and return it as the answer to the query.
\item
\textbf{Case I-b: $\pmb{\kappa_{x_i^j,s}=\kappa_{i,s}}$}. 
We return an invalid entry 
as the answer to the query.
\end{enumerate}
\item
\textbf{Case II: $\pmb{\kappa_{i,s}<\sigma_{i,s}}$}. 
We make a new query 
$(x_i,s)$ on $G_{u,v}$
giving us a node $y_p\in \Lambda_{i,s}=\cS_{x_i,s}\setminus \cT_{i,s}$
with the property that 
$
\Pr [ y_p\in \Lambda_{i,s} \text{ is returned}\!]
=
\frac{1}{
\lambda_{i,s}
}
$.
\begin{enumerate}[label=$\blacktriangleright$,leftmargin=*]
\item
\textbf{Case II-a: $\pmb{\kappa_{x_i^j,s} = \kappa_{i,s}}$}.
For this case,
$
\cT_{i,s}
=
\cT_{x_i^j,s}
$
and 
$\lambda_{i,s}=\Lambda_{x_i^j,s}$.
We return the 
node $y_p$ 
as the answer to the query and 
update all relevant sets and counters appropriately.
\item
\textbf{Case II-b: $\pmb{\kappa_{x_i^j,s} < \kappa_{i,s}}$}.
For this case
$\cT_{x_i^j,s} \subset \cT_{i,s} \subset \cS_{x_i,s}$,
$\lambda_{i,s}= |\cS_{x_i,s}\setminus \cT_{i,s}|>0$, and 
$
\lambda_{x_i^j,s} = | \cS_{x_i,s} \setminus \cT_{x_i^j,s} | > \lambda_{i,s}
$.
We sample the nodes in 
$
\{ y_p \} \bigcup \big( \cT_{i,s} \setminus \cT_{x_i^j,s} \big)
$
based on the following probability distribution 
and update all relevant sets and counters appropriately:
\begin{gather*}
\textstyle
\Pr [y_p \text{ is selected}\!] = \frac{\lambda_{i,s}}{\lambda_{x_i^j,s}}
\\
\textstyle
\forall \, y_\ell\in \cT_{i,s} \setminus \cT_{x_i^j,s} \,:\,
\Pr[ y_\ell \text{ is selected}\!] = \frac{1}{\lambda_{x_i^j,s}}
\end{gather*}
Thus, 
the answer to the query
$(x_i^j,s)$
is selected \emph{uniformly at random} from the set 
$\Lambda_{x_i^j,s}=\cS_{x_i,s}\setminus \cT_{x_i^j,s}$ since 
\begin{gather*}
\begin{array}{r l}
\textstyle
\forall\, y_\ell \in  \cS_{x_i,s}\setminus \cT_{i,s} \, : &
\textstyle
\Pr[ y_\ell \text{ is selected}\!] = \frac{1}{\lambda_{i,s}}\times\frac{\lambda_{i,s}}{\lambda_{x_i^j,s}} =
		 = \frac{1}{\lambda_{x_i^j,s} }
\\
\textstyle
\forall \, y_\ell\in \cT_{i,s} \setminus \cT_{x_i^j,s} \, : & 
\textstyle
\Pr[ y_\ell \text{ is selected}\!] 
		 = \frac{1}{\lambda_{x_i^j,s} }
\end{array}
\end{gather*}
\end{enumerate}
\end{enumerate}
\end{description}
\item
\textbf{Case 3.2: 
$\pmb{q_i'}$ is of the form 
$\pmb{(r_i,s)}$}.
We keep a count 
$\nu(r_i)$
of how many times the query 
$(r_i,3)$
has been asked involving the node 
$r_i$
before the current query, 
and store the answers to these queries in the set 
$\cS_{r_i}$.
If $s\neq 3$ 
or 
$\nu(r_i)=\deg_G(v)+1$
we return the special symbol.
Otherwise, 
we return a node 
selected uniformly at random from the set of nodes 
$\{ y_1,\dots,y_{\deg_G(v)+1} \} 
\setminus 
\cS_{r_i}
$ 
as the answer 
and update all relevant sets and counters appropriately.
\item
\textbf{Case 3.3: 
$\pmb{q_i'}$ is of the form 
$\pmb{(y_\ell,s)}$}.
This case is similar in spirit to  
Case 3.1.
We show how 
to handle all queries of the form 
$(y_\ell,s)$
\textbf{for each specific} $\ell$ and $s$.
\begin{description}
\item[(\emph{i})]
Assume without loss of generality that 
$y_\ell$ 
is connected,
via edges of weight $s$,
to (not initially known to $\bbbless$)
a set
$\cS_{\nu_1}=\{x_1,\dots,x_{\nu_1}\}\subseteq \{x_1,\dots,x_{\deg_G(u)+1}\}$
of $\nu_1=|\cS_{\nu_1}|$ nodes. 
If \emph{not} already done before, we make one new weighted selective degree query 
$(y_\ell,s)$ 
on $G_{u,v}$ 
giving us the value of 
$\nu_1$
(if $\nu_1$ \emph{is} already known we simply use it without making a query).
\item[(\emph{ii})]
Define the set $\cS_{\nu_2}$ 
of $\nu_2=|\cS_{\nu_2}|\in\{0,b\}$ nodes
as 
$\cS_{\nu_2}=\{r_1,\dots,r_b\}$ if $s=3$ and $\cS_{\nu_2}=\emptyset$ otherwise.
Note that we know the value of $\nu_2$ since we know the value of $s$.
\item[(\emph{iii})]
We keep a count 
$\kappa$
of how many times the query 
$(y_\ell,s)$
has been asked involving the node 
$y_\ell$ 
before the current query, and let 
$\cT_\kappa$ be the set of those $\kappa=|\cT_\kappa|$
nodes \textbf{of} $\pmb{G_{u,v}'}$ that have been returned because of these prior queries.
\textbf{Note that if 
$\pmb{\kappa<a\,\nu_1+\nu_2}$ 
then 
performing a new weighted neighbor query 
$\pmb{(y_\ell,s)}$ on $\pmb{G_{u,v}'}$
\emph{must} return
a node uniformly at random from 
the set of nodes $\pmb{\Lambda_\kappa = 
\big ( \cup_{i=1}^{\nu_1} \cup_{j=1}^a \{x_i^j\} \cup 
\cS_{\nu_2}
\big) 
\setminus \cT_\kappa
}$
with probability $\pmb{\nicefrac{1}{\lambda_\kappa}}$ where $\pmb{\lambda_\kappa = |\Lambda_\kappa| = 
(a\,\nu_1+\nu_2)-\kappa}$}.
\item[(\emph{iv})]
Assume without loss of generality that 
$\cS_{\nu_1}'=\{x_1,x_2,\dots,x_{\nu_1'}\}\subseteq \cS_{\nu_1}$
be the set of $\nu_1'=|\cS_{\nu_1}'|\leq \min \{\kappa, \nu_1 \} \}$ nodes 
\textbf{in} $\pmb{G_{u,v}}$ that have been 
returned as a result of the queries on $G_{u,v}$ 
due to the simulation of prior 
$\kappa$
queries on $G_{u,v}'$.
Note that 
\textbf{if $\pmb{\nu_1'<\nu_1}$ then 
performing a new weighted neighbor query 
$\pmb{(y_\ell,s)}$ on $\pmb{G_{u,v}}$
returns a \emph{new} node 
uniformly at random
from the set of nodes 
$\pmb{\Phi=\cS_{\nu_1} \setminus \cS_{\nu_1}'}$
with probability $\pmb{\nicefrac{1}{\varphi}}$ where 
$\pmb{\varphi=|\Phi|=\nu_1-\nu_1'}$}.
\item[(\emph{v})]
Define the subset $\Lambda_\kappa'\subseteq\Lambda_\kappa$ of nodes of $G_{u,v}'$ 
as 
$
\Lambda_\kappa' = \big( \cup_{i=1}^{\nu_1'} \cup_{j=1}^a \{x_i^j\} \cup \cS_{\nu_2} \big) \setminus \cT_\kappa
$.
Note that we know all the elements of $\Lambda_\kappa'$ and 
$
\lambda_\kappa' = |\Lambda_\kappa'| = (a \nu_1' + \nu_2 ) - \kappa 
$.
In particular, this means that \textbf{we \emph{can} sample a node from $\pmb{\Lambda_\kappa'}$ uniformly at random}.
\item[(\emph{vi})]
For a new query 
$(y_\ell,s)$, 
we have the following cases.
\begin{enumerate}[label=$\blacktriangleright$,leftmargin=*]
\item
\textbf{Case I: $\pmb{\kappa\geq\nu_1}$}. 
In this case, $\nu_1'=\nu_1$, $\Lambda_\kappa'=\Lambda_\kappa$ and $\lambda_\kappa'=\lambda_\kappa$.
We select as our answer to the query a node uniformly at random from 
$\Lambda_\kappa'$,
and update all relevant sets and counters appropriately.
\item
\textbf{Case II: $\pmb{\kappa<\nu_1}$}. 
In this case, 
$\nu_1'<\nu_1$,
and 
$\varphi>0$. We simulate the query as follows.
\begin{itemize}
\item
We make a new query 
$(y_\ell,s)$ on $G_{u,v}$
giving us a node
$x_{p}\in\Phi$ for $p\in\{\nu_1'+1,\dots,\nu_1\}$ 
with probability $\nicefrac{1}{\varphi}$.
We select $j\in\{1,\dots,a\}$ uniformly at random giving us a node $x_p^j$.
\item
We \emph{sample} a node from $\{x_{p}^j\}\cup\Lambda_\kappa'$ based on the following probability distribution
and update all relevant sets and counters appropriately:
\begin{gather*}
\textstyle
\Pr[ x_p^j \text{ is selected } \! ] = \frac{a \varphi}{ \lambda_\kappa}
\\
\textstyle
\forall \, x_i^j \in \Lambda_\kappa' \,:\,
\Pr[ x_i^j \text{ is selected } \! ] = \frac{1}{\lambda_\kappa}
\end{gather*}
Note that 
$
\Pr[ x_i^j\in  \Lambda_\kappa \setminus \Lambda_\kappa'  \text{ is selected } \! ] 
=
\frac{1}{a \varphi} \times \frac{a \varphi}{ \lambda_\kappa}
=
\frac{1}{ \lambda_\kappa}
$, 
as desired.
To verify that all the probabilities add up to $1$, note that 
$
\Pr[ x_p^j  \text{ is selected } \! ] 
+
\sum_{x_i^j \in \Lambda_\kappa'}
\Pr[ x_i^j \in \Lambda_\kappa' \text{ is selected } \! ] 
=
\frac { a (\nu_1 - \nu_1') } { a\nu_1 + \nu_2 - \kappa }
+
(a\nu_1'+\nu_2-\kappa) \times \frac{1} { a\nu_1 + \nu_2 - \kappa}
=1
$.
\end{itemize}
\end{enumerate}
\end{description}
\end{enumerate}
\end{enumerate}
\end{proof}

\newcommand{\bbbe}{\mathfrak{B}}
\newcommand{\rr}{\mathfrak{r}}

\section{Computing $\mathfrak{C}_{G}(v)$ and $\mathfrak{C}_{\mathrm{avg}}(G)$
using ``black box'' additive approximation algorithms for $\mathfrak{C}_{G}(e)$}
\label{sec-ricci-nodes}

In this section we provide efficient local algorithms to compute 
$\mathfrak{C}_{G}(v)$ and 
$\mathfrak{C}_{\mathrm{avg}}(G)$.
The following assumptions are used by our algorithms:
\begin{enumerate}[label=$\triangleright$]
\item
For a fixed $\rr$,
we have an efficient local algorithm $\bbbe$ for an additive $\rr$-approximation, say 
$\mathfrak{C}_{G}^{\,\rr}(e)$, 
of $\mathfrak{C}_{G}(e)$ for an edge $e$.
\item
We have access to the \emph{neighbor query model} mentioned in Section~\ref{sec-query-defn}.
\end{enumerate}

\begin{lemma}\label{lem-ave}
With probability at least $\nicefrac{2}{3}$ the following two claims hold.
\begin{enumerate}[label=\emph{(\emph{\alph*})}]
\item
We can compute an additive $2\rr$-approximation of 
$\mathfrak{C}_{G}(v)$
using $O(1/\rr^2)$ neighbor queries and 
$O(1/\rr^2)$ 
invocations of algorithm $\bbbe$
on the edges incident on $v$,
and 
\item
If the degrees of all the nodes of $G$ are know then 
we can compute an additive $2\rr$-approximation of 
$\mathfrak{C}_{\mathrm{avg}}(G)$
using $O(1/\rr^2)$ neighbor queries and 
$O(1/\rr^2)$ 
invocations of algorithm $\bbbe$
over all edges in $G$.
\end{enumerate}
\end{lemma}

\begin{proof}~\\

\vspace*{-0.1in}
\noindent
\textbf{(\emph{a})}
Let $k$ be a parameter to be specified later. 
We use $k'=\min\{k,\deg_G(v)\}$ neighbor queries to get $k'$ nodes adjacent to $v$, say $u_1,\dots,u_{k'}$, 
compute 
$
\mathfrak{C}_{G}^{\,\rr}(v,u_1), 
\dots,
\mathfrak{C}_{G}^{\,\rr}(v,u_{k'}) 
$
using algorithm $\bbbe$,
and return 
$
\widetilde{\mathfrak{C}_{G}}(v)= 
\frac{1}{k'}\sum_{j=1}^{k'} 
\mathfrak{C}_{G}^{\,\rr}(v,u_j) 
$
as our answer.

If $k>\deg_G(v)$ then 
$\mathfrak{C}_{G}(v)= 
\frac{1}{k'}\sum_{j=1}^{k'} 
\mathfrak{C}_{G}(v,u_j) 
$ 
and thus 
$\widetilde{\mathfrak{C}_{G}}(v)$ is in fact an additive $\rr$-approximation of 
$\mathfrak{C}_{G}(v)$. Otherwise, assume that 
$k\leq \deg_G(v)$ and therefore $k'=k$.
For any number $x$, we use the notation $x\boxplus \rho$
to indicate a number $y$ that satisfies $x\leq y\leq x+\rho$.
Observe that
\begin{multline*}
\Ave{\widetilde{\mathfrak{C}_{G}}(v)}
=
\frac{1}{k}\sum_{j=1}^{k} 
\Ave{\mathfrak{C}_{G}^{\,\rr}(v,u_j)} 
=
\frac{1}{k}\sum_{j=1}^{k} 
\sum_{u \in \nbr_G(u)} (\mathfrak{C}_{G}((u,v))  \boxplus \rr )\times \frac{1}{\deg_G(v)}
\\
=
\frac{1}{k}\sum_{j=1}^{k} 
( \mathfrak{C}_{G}(v) \boxplus \rr )
=
\mathfrak{C}_{G}(v) \boxplus \rr
\end{multline*}
Since 
the $\frac{1}{k}\mathfrak{C}_{G}^{\,\rr}(v,u_j)$'s 
are mutually independent for $j=1,\dots,k$, and each 
$\frac{1}{k}\mathfrak{C}_{G}^{\,\rr}(v,u_j)$
lies in the interval 
$[-2/k,1/k]$ (\emph{cf}.\ see Section~\ref{sec-simple-bo}), 
applying Hoeffding's inequality~\cite[Theorem 2]{hoeffding63} we get
\begin{gather*}
\Pr[\widetilde{\mathfrak{C}_{G}}(v) > \mathfrak{C}_{G}(v) + 2\rr]
\leq
\Pr[\widetilde{\mathfrak{C}_{G}}(v) >  \eX{\widetilde{\mathfrak{C}_{G}}(v)} + \rr]
\\
\textstyle
<
\exp \left( - \frac{2\rr^2} { \sum^k_{i=1} (2/k - (-1/k) )^2 } \right)
=
\exp \left(  - \frac{2}{9} k^2 \rr^2 \right)
\end{gather*}
Setting $k=\Theta(\rr^{-2})$
we get  
$\Pr[\widetilde{\mathfrak{C}_{G}}(v) > \mathfrak{C}_{G}(v) + 2\rr] < \nicefrac{1}{3}$.

\medskip
\noindent
\textbf{(\emph{b})}
The algorithm and its proof is very similar to those in 
\textbf{(\emph{a})}.
For this case, we need to randomly sample 
$k'=\min\{ O(1/\rr^2), ,|E|\}$ 
edges $e_1,\dots,e_{k'}$ from $E$, compute 
$
\mathfrak{C}_{G}^{\,\rr}(e_1), 
\dots,
\mathfrak{C}_{G}^{\,\rr}(e_{k'}) 
$
using algorithm $\bbbe$,
and return 
$
\frac{1}{k'}\sum_{j=1}^{k'} 
\mathfrak{C}_{G}^{\,\rr}(e_j) 
$
as our answer.
The only remaining part of the proof is to show how to sample an edge uniformly at random 
from the set of edges $E$ of $G$. Since the degrees of all nodes are known, the following procedure 
can be used. We first select a node $x\in V$ with probability 
$
\frac { \deg_G(x) } { \sum_{z\in V} \deg_G(z) }
$, then we select a random neighbor of $x$, say $y$, using one 
neighbor query, and finally we select the edge $\{x,y\}$. 
The proof is completed by observing that 
\begin{multline*}
\textstyle
\Pr[\{u,v\}\in E \text{ is selected} ]
\\
\textstyle
=
\Pr [ x\in V \text{ is selected} ] \times \Pr [ y\in\nbr_G(x) \text{ is selected} ]
\\
\textstyle
\hspace*{0.3in}
+
\Pr [ y\in V \text{ is selected} ] \times \Pr [ x\in\nbr_G(y) \text{ is selected} ]
\\
\textstyle
=
\frac { \deg_G(x) } { \sum_{z\in V} \deg_G(z) } \times \frac {1} { \deg_G(x) } + 
\frac { \deg_G(y) } { \sum_{z\in V} \deg_G(z) } \times \frac {1} { \deg_G(y) }
=
\frac{1}{|E|}
\end{multline*}
\end{proof}

\section{Concluding remarks}
\label{sec-conclude}

We hope that this paper will stimulate further attention from computer scientists
concerning the exciting interplay between notions of 
curvatures from network and non-network domains.
An obvious candidate for future research is improvement of the query complexities for local algorithms for
computing the Ollivier-Ricci curvature for networks.
Another possible future research direction is to investigate computational complexity issues of other discretizations of Ricci curvatures. 
For example, another discretization of Ricci curvature for networks proposed by 
Ollivier and Villani~\cite{OllV12}
is guided by the observation that the 
infinite-dimensional version of the
well-known {Brunn-Minkowski inequality} over $\R^n$~\cite{Ga02} 
can be {tightened} in the presence of 
a positive curvature for a smooth Riemannian manifold~\cite{cms01,cms06}.
To our knowledge, these discretizations have largely escaped computational complexity considerations.

\appendix

\section{A self-contained proof of Fact~\ref{fact2}}

Let $\deg_G(u)=\deg_G(v)=\alpha$.
Build a directed single-source single-sink flow network~\cite{PS82} 
$G_{u,v}^f$ from $G_{u,v}$ in the following manner: add a new source node $s$ and a new sink node $t$, 
add an arc (directed edge) from $s$ to every node of $L_{u,v}^G$ of weight zero and capacity $1$, 
add an arc from every node of $R_{u,v}^G$ to $t$ of weight zero and capacity $1$, 
orient every edge $\{{u'},{v'}\}$ of $G_{u,v}$ from ${u'}$ to ${v'}$ and set its capacity to $1$. 
Since $|L_{u,v}^G|=|R_{u,v}^G|=\alpha+1$, 
we have $\PP_{u}^G(u')=\PP_{v}^G(v')=\frac{1}{\alpha+1}$
for all 
$u'\in \nbr^G(u)\cup \{u\}$
and 
$v'\in \nbr^G(v)\cup \{v\}$.
Thus, since $G_{u,v}$ is a complete bipartite graph, 
by a simple scaling it follows that 
$
\text{\emd}_{\!\!\!G_{u,v}} (\PP_u^G,\PP_v^G)
=\frac{\cM}{\alpha+1}$ 
where $\cM$ is the total weight of a minimum-weight maximum $s$-$t$ flow 
on $G_{u,v}^f$. 
Since the node-arc incidence matrix of a directed graph is totally unimodular, 
the flow value of every arc of any extreme-point optimal solution 
of the minimum-weight maximum $s$-$t$ flow 
on $G_{u,v}^f$ is integral and therefore $0$ or $1$ (see Theorem $13.3$ and its corollary in~\cite{PS82}). 
This integrality of flow values and the fact that 
$G_{u,v}$ is a complete bipartite graph
imply $\cM$ is also the total weight of a minimum-weight \emph{perfect} matching of 
$G_{u,v}$.

We now show that there is such a minimum-weight perfect matching that uses all the zero-weight edges
$\{{u'},{u'}\}$
for all 
$u'\in \{u,v\} \cup \big( \,\nbr_G(u)\cap \nbr_G(v)  \,\big)$.
For a contradiction, suppose that the edge 
$\{{u'},{u'}\}$
is not used for some 
$u'\in \{u,v\} \cup ( \nbr_G(u)\cap \nbr_G(v)  ) \}$.
Since our solution is a perfect matching, the nodes 
${u'}\in L_{u,v}^G$ and ${u'}\in R_{u,v}^G$
must be matched to some other nodes, say to nodes 
${v''}\in R_{u,v}^G$
and 
$u''\in L_{u,v}^G$, respectively.
Then, if we instead use the edges 
$\{{u'},{u'}\}$
and 
$\{{u''},{v''}\}$
then using 
the triangle inequality 
it follows that 
the total weight of this modified perfect matching is no more than that of the original perfect matching since:
\begin{multline*}
w_{u,v}^G({u'},{u'}) + w_{u,v}^G({u''},{v''}) = 
w_{u,v}^G({u''},{v''}) \leq  
w_{u,v}^G({u''},{u'}) + 
w_{u,v}^G({u'},{u'}) + 
w_{u,v}^G({u'},{v''})
\\
=
w_{u,v}^G({u''},{u'}) + 
w_{u,v}^G({u'},{v''})
\end{multline*}


\begin{thebibliography}{10}
\expandafter\ifx\csname url\endcsname\relax
  \def\url#1{\texttt{#1}}\fi
\expandafter\ifx\csname urlprefix\endcsname\relax\def\urlprefix{URL }\fi
\expandafter\ifx\csname href\endcsname\relax
  \def\href#1#2{#2} \def\path#1{#1}\fi

\bibitem{book99}
M.~R. Bridson, A.~H\"{a}fliger, Metric Spaces of Non-Positive Curvature, 1st
  Edition, Springer-Verlag Berlin Heidelberg, 1999.
\newblock \href {https://doi.org/10.1007/978-3-662-12494-9}
  {\path{doi:10.1007/978-3-662-12494-9}}.

\bibitem{Berger12}
M.~Berger, A Panoramic View of Riemannian Geometry, 1st Edition,
  Springer-Verlag Berlin Heidelberg, 2003.
\newblock \href {https://doi.org/10.1007/978-3-642-18245-7}
  {\path{doi:10.1007/978-3-642-18245-7}}.

\bibitem{ADM14}
R.~{Albert}, B.~{DasGupta}, N.~Mobasheri,
  \href{https://link.aps.org/doi/10.1103/PhysRevE.89.032811}{Topological
  implications of negative curvature for biological and social networks},
  Physical Review E 89 (2014) 032811.
\newblock \href {https://doi.org/10.1103/PhysRevE.89.032811}
  {\path{doi:10.1103/PhysRevE.89.032811}}.
\newline\urlprefix\url{https://link.aps.org/doi/10.1103/PhysRevE.89.032811}

\bibitem{CATAD21}
T.~Chatterjee, R.~{Albert}, S.~Thapliyal, N.~Azarhooshang, B.~{DasGupta},
  Detecting network anomalies using forman-ricci curvature and a case study for
  human brain networks, Scientific Reports 11 (2021).
\newblock \href {https://doi.org/10.1038/s41598-021-87587-z}
  {\path{doi:10.1038/s41598-021-87587-z}}.

\bibitem{JLBB11}
E.~Jonckheere, M.~Lou, F.~Bonahon, Y.~Baryshnikov,
  \href{https://doi.org/10.1080/15427951.2010.554320}{Euclidean versus
  hyperbolic congestion in idealized versus experimental networks}, Internet
  Mathematics 71 (2011) 1--27.
\newblock \href {https://doi.org/10.1080/15427951.2010.554320}
  {\path{doi:10.1080/15427951.2010.554320}}.
\newline\urlprefix\url{https://doi.org/10.1080/15427951.2010.554320}

\bibitem{SJB21}
J.~Sia, E.~Jonckheere, P.~Bogdan, Ollivier-ricci curvature-based method to
  community detection in complex networks, Scientific Reports 9 (2019) 9800.
\newblock \href {https://doi.org/10.1038/s41598-019-46079-x}
  {\path{doi:10.1038/s41598-019-46079-x}}.

\bibitem{Ricci}
A.~K. Simhal, K.~L.~H. Carpenter, S.~Nadeem, J.~Kurtzberg, A.~Song,
  A.~Tannenbaum, G.~Sapiro, G.~Dawson, Measuring robustness of brain networks
  in autism spectrum disorder with \mbox{R}icci curvature, Scientific Reports
  10 (2020) 10819.
\newblock \href {https://doi.org/10.1038/s41598-020-67474-9}
  {\path{doi:10.1038/s41598-020-67474-9}}.

\bibitem{Elumalai2021}
P.~Elumalai, Y.~Yadav, N.~Williams, E.~Saucan, J.~Jost, A.~Samal,
  \href{https://www.biorxiv.org/content/early/2021/12/21/2021.11.28.470231}{Graph
  ricci curvatures reveal atypical functional connectivity in autism spectrum
  disorder}, bioRxiv (2021).
\newblock \href {https://doi.org/10.1101/2021.11.28.470231}
  {\path{doi:10.1101/2021.11.28.470231}}.
\newline\urlprefix\url{https://www.biorxiv.org/content/early/2021/12/21/2021.11.28.470231}

\bibitem{CL03}
B.~Chow, F.~Luo, Combinatorial ricci flows on surfaces, Journal of Differential
  Geometry 63~(1) (2003) 97--129.
\newblock \href {https://doi.org/10.4310/jdg/1080835659}
  {\path{doi:10.4310/jdg/1080835659}}.

\bibitem{Oll11}
Y.~Ollivier, \href{https://hal.archives-ouvertes.fr/hal-00858008}{A visual
  introduction to \mbox{R}iemannian curvatures and some discrete
  generalizations}, in: G.~Dafni, R.~J. McCann, A.~Stancu (Eds.), Analysis and
  Geometry of Metric Measure Spaces: Lecture Notes of the 50th S{\'e}minaire de
  Math{\'e}matiques Sup{\'e}rieures (SMS), Montr\'{e}al, 2011, Vol.~56,
  American Mathematical Society, Providence, RI, USA, 2013, pp. 197--219.
\newblock \href {https://doi.org/10.1090/crmp/056/08}
  {\path{doi:10.1090/crmp/056/08}}.
\newline\urlprefix\url{https://hal.archives-ouvertes.fr/hal-00858008}

\bibitem{Oll09}
Y.~Ollivier, Ricci curvature of markov chains on metric spaces, Journal of
  Functional Analysis 256 (2009) 810--864.
\newblock \href {https://doi.org/10.1016/j.jfa.2008.11.001}
  {\path{doi:10.1016/j.jfa.2008.11.001}}.

\bibitem{Oll10}
Y.~Ollivier, A survey of ricci curvature for metric spaces and markov chains,
  in: M.~Kotani, M.~Hino, T.~Kumagai (Eds.), Advanced Studies in Pure
  Mathematics, Vol.~57, Mathematical Society of Japan, 2010, pp. 343--381.
\newblock \href {https://doi.org/10.2969/aspm/05710343}
  {\path{doi:10.2969/aspm/05710343}}.

\bibitem{Oll07}
Y.~Ollivier,
  \href{https://www.sciencedirect.com/science/article/pii/S1631073X07004414}{Ricci
  curvature of metric spaces}, Comptes Rendus Mathematique 345~(11) (2007)
  643--646.
\newblock \href {https://doi.org/10.1016/j.crma.2007.10.041}
  {\path{doi:10.1016/j.crma.2007.10.041}}.
\newline\urlprefix\url{https://www.sciencedirect.com/science/article/pii/S1631073X07004414}

\bibitem{DJY20}
B.~{DasGupta}, M.~V. Janardhanan, F.~Yahyanejad, Why did the shape of your
  network change? (on detecting network anomalies via non-local curvatures),
  Algorithmica 82~(7) (2020) 1741--1783.
\newblock \href {https://doi.org/10.1007/s00453-019-00665-7}
  {\path{doi:10.1007/s00453-019-00665-7}}.

\bibitem{DKMF18}
B.~{DasGupta}, M.~Karpinski, N.~Mobasheri, F.~Yahyanejad, Effect of
  gromov-hyperbolicity parameter on cuts and expansions in graphs and some
  algorithmic implications, Algorithmica 80~(2) (2018) 772--800.
\newblock \href {https://doi.org/10.1007/s00453-017-0291-7}
  {\path{doi:10.1007/s00453-017-0291-7}}.

\bibitem{a1}
I.~Benjamini, Expanders are not hyperbolic, Israel Journal of Mathematics 108
  (1998) 33--36.
\newblock \href {https://doi.org/10.1007/BF02783040}
  {\path{doi:10.1007/BF02783040}}.

\bibitem{CCDDMV18}
J.~Chalopin, V.~Chepoi, F.~F. Dragan, G.~Ducoffe, A.~M. A., Y.~Vax\`{e}s, Fast
  approximation and exact computation of negative curvature parameters of
  graphs., Discrete and Computational Geometry 65 (2021) 856--892.
\newblock \href {https://doi.org/10.1007/s00454-019-00107-9}
  {\path{doi:10.1007/s00454-019-00107-9}}.

\bibitem{ipl15}
H.~Fournier, A.~Ismail, A.~Vigneron,
  \href{https://doi.org/10.1016/j.ipl.2015.02.002}{Computing the gromov
  hyperbolicity of a discrete metric space}, Information Processing Letters
  115~(6) (2015) 576--579.
\newblock \href {https://doi.org/10.1016/j.ipl.2015.02.002}
  {\path{doi:10.1016/j.ipl.2015.02.002}}.
\newline\urlprefix\url{https://doi.org/10.1016/j.ipl.2015.02.002}

\bibitem{F03}
R.~Forman, Bochner's method for cell complexes and combinatorial ricci
  curvature, Discrete and Computational Geometry 29~(3) (2003) 323--374.
\newblock \href {https://doi.org/10.1007/s00454-002-0743-x}
  {\path{doi:10.1007/s00454-002-0743-x}}.

\bibitem{Sree1}
R.~P. Sreejith, K.~Mohanraj, J.~Jost, E.~Saucan, A.~Samal, Forman curvature for
  complex networks, Journal of Statistical Mechanics: Theory and Experiment
  2016~(6) (2016) 063206.
\newblock \href {https://doi.org/10.1088/1742-5468/2016/06/063206}
  {\path{doi:10.1088/1742-5468/2016/06/063206}}.

\bibitem{Sree2}
R.~P. Sreejith, J.~Jost, E.~Saucan, A.~Samal,
  \href{https://www.sciencedirect.com/science/article/pii/S0960077917302102}{Systematic
  evaluation of a new combinatorial curvature for complex networks}, Chaos,
  Solitons and Fractals 101 (2017) 50--67.
\newblock \href {https://doi.org/10.1016/j.chaos.2017.05.021}
  {\path{doi:10.1016/j.chaos.2017.05.021}}.
\newline\urlprefix\url{https://www.sciencedirect.com/science/article/pii/S0960077917302102}

\bibitem{Weber17}
M.~Weber, E.~Saucan, J.~Jost, Characterizing complex networks with forman-ricci
  curvature and associated geometric flows, Journal of Complex Networks 5~(4)
  (2017) 527--550.
\newblock \href {https://doi.org/10.1093/comnet/cnw030}
  {\path{doi:10.1093/comnet/cnw030}}.

\bibitem{Samal18}
A.~Samal, R.~P. Sreejith, J.~Gu, S.~Liu, E.~Saucan, J.~Jost, Comparative
  analysis of two discretizations of ricci curvature for complex networks,
  Scientific Reports 8 (2018) 8650.
\newblock \href {https://doi.org/10.1038/s41598-018-27001-3}
  {\path{doi:10.1038/s41598-018-27001-3}}.

\bibitem{G87}
M.~Gromov, Hyperbolic groups, in: S.~M. Gersten (Ed.), Essays in Group Theory,
  Vol.~8, Springer, New York, NY, 1987, pp. 75--263.
\newblock \href {https://doi.org/10.1007/978-1-4613-9586-7\_3}
  {\path{doi:10.1007/978-1-4613-9586-7\_3}}.

\bibitem{CDEHV08}
V.~Chepoi, F.~Dragan, B.~Estellon, M.~Habib, Y.~Vax\`{e}s,
  \href{https://doi.org/10.1145/1377676.1377687}{Diameters, centers, and
  approximating trees of delta-hyperbolicgeodesic spaces and graphs}, in:
  Proceedings of the Twenty-Fourth Annual Symposium on Computational Geometry,
  SCG '08, Association for Computing Machinery, New York, NY, USA, 2008, pp.
  59--68.
\newblock \href {https://doi.org/10.1145/1377676.1377687}
  {\path{doi:10.1145/1377676.1377687}}.
\newline\urlprefix\url{https://doi.org/10.1145/1377676.1377687}

\bibitem{PKBV10}
F.~Papadopoulos, D.~Krioukov, M.~Boguna, A.~Vahdat, Greedy forwarding in
  dynamic scale-free networks embedded in hyperbolic metric spaces, in: 2010
  Proceedings IEEE INFOCOM, 2010, pp. 1--9.
\newblock \href {https://doi.org/10.1109/INFCOM.2010.5462131}
  {\path{doi:10.1109/INFCOM.2010.5462131}}.

\bibitem{YYI2012}
Y.~Yoshida, M.~Yamamoto, H.~Ito, {Improved Constant-Time Approximation
  Algorithms for Maximum Matchings and Other Optimization Problems}, SIAM
  Journal on Computing 41~(4) (2012) 1074--1093.
\newblock \href {https://doi.org/10.1137/110828691}
  {\path{doi:10.1137/110828691}}.

\bibitem{LS15}
Y.~T. Lee, A.~Sidford, Efficient inverse maintenance and faster algorithms for
  linear programming, in: 2015 IEEE 56th Annual Symposium on Foundations of
  Computer Science, 2015, pp. 230--249.
\newblock \href {https://doi.org/10.1109/FOCS.2015.23}
  {\path{doi:10.1109/FOCS.2015.23}}.

\bibitem{Q18}
K.~Quanrud,
  \href{http://drops.dagstuhl.de/opus/volltexte/2018/10032}{{Approximating
  Optimal Transport With Linear Programs}}, in: J.~T. Fineman, M.~Mitzenmacher
  (Eds.), 2nd Symposium on Simplicity in Algorithms (SOSA 2019), Vol.~69 of
  OpenAccess Series in Informatics (OASIcs), Schloss Dagstuhl--Leibniz-Zentrum
  fuer Informatik, Dagstuhl, Germany, 2018, pp. 6:1---6:9.
\newblock \href {https://doi.org/10.4230/OASIcs.SOSA.2019.6}
  {\path{doi:10.4230/OASIcs.SOSA.2019.6}}.
\newline\urlprefix\url{http://drops.dagstuhl.de/opus/volltexte/2018/10032}

\bibitem{DGK18}
P.~Dvurechensky, A.~Gasnikov, A.~Kroshnin,
  \href{https://proceedings.mlr.press/v80/dvurechensky18a.html}{Computational
  optimal transport: Complexity by accelerated gradient descent is better than
  by sinkhorn's algorithm}, in: J.~Dy, A.~Kraus (Eds.), Proceedings of the 35th
  International Conference on Machine Learning, Vol.~80 of Proceedings of
  Machine Learning Research, PMLR, 2018, pp. 1367---1376.
\newline\urlprefix\url{https://proceedings.mlr.press/v80/dvurechensky18a.html}

\bibitem{ASD20}
N.~Azarhooshang, P.~Sengupta, B.~{DasGupta}, A review of and some results for
  ollivier-ricci network curvature, Mathematics 8~(1416) (2020).
\newblock \href {https://doi.org/10.3390/math8091416}
  {\path{doi:10.3390/math8091416}}.

\bibitem{PC19}
G.~Peyr\'{e}, M.~Cuturi,
  \href{http://dx.doi.org/10.1561/2200000073}{Computational optimal transport:
  With applications to data science}, Foundations and Trends in Machine
  Learning 11~(5--6) (2019) 355--607.
\newblock \href {https://doi.org/10.1561/2200000073}
  {\path{doi:10.1561/2200000073}}.
\newline\urlprefix\url{http://dx.doi.org/10.1561/2200000073}

\bibitem{GS02}
A.~L. Gibbs, F.~E. Su, \href{http://www.jstor.org/stable/1403865}{On choosing
  and bounding probability metrics}, International Statistical Review / Revue
  Internationale de Statistique 70~(3) (2002) 419--435.
\newblock \href {https://doi.org/10.2307/1403865} {\path{doi:10.2307/1403865}}.
\newline\urlprefix\url{http://www.jstor.org/stable/1403865}

\bibitem{VVW19}
V.~V. Williams, On some fine-grained questions in algorithms and complexity,
  in: Proceedings of the International Congress of Mathematicians (ICM 2018),
  2019, pp. 3447--3487.
\newblock \href {https://doi.org/10.1142/9789813272880\_0188}
  {\path{doi:10.1142/9789813272880\_0188}}.

\bibitem{AGW15}
A.~Abboud, F.~Grandoni, V.~V. Williams, Subcubic equivalences between graph
  centrality problems, apsp and diameter, in: Proceedings of the Twenty-Sixth
  Annual ACM-SIAM Symposium on Discrete Algorithms, SODA '15, Society for
  Industrial and Applied Mathematics, USA, 2015, pp. 1681--1697.

\bibitem{Pat10}
M.~Patrascu, \href{https://doi.org/10.1145/1806689.1806772}{Towards polynomial
  lower bounds for dynamic problems}, in: Proceedings of the Forty-Second ACM
  Symposium on Theory of Computing, STOC '10, Association for Computing
  Machinery, New York, NY, USA, 2010, pp. 603--610.
\newblock \href {https://doi.org/10.1145/1806689.1806772}
  {\path{doi:10.1145/1806689.1806772}}.
\newline\urlprefix\url{https://doi.org/10.1145/1806689.1806772}

\bibitem{Lee02}
L.~Lee, \href{https://doi.org/10.1145/505241.505242}{Fast context-free grammar
  parsing requires fast boolean matrix multiplication}, Journal of the ACM
  49~(1) (2002) 1--15.
\newblock \href {https://doi.org/10.1145/505241.505242}
  {\path{doi:10.1145/505241.505242}}.
\newline\urlprefix\url{https://doi.org/10.1145/505241.505242}

\bibitem{PARNAS2007}
M.~Parnas, D.~Ron,
  \href{https://www.sciencedirect.com/science/article/pii/S0304397507003696}{Approximating
  the minimum vertex cover in sublinear time and a connection to distributed
  algorithms}, Theoretical Computer Science 381~(1) (2007) 183--196.
\newblock \href {https://doi.org/https://doi.org/10.1016/j.tcs.2007.04.040}
  {\path{doi:https://doi.org/10.1016/j.tcs.2007.04.040}}.
\newline\urlprefix\url{https://www.sciencedirect.com/science/article/pii/S0304397507003696}

\bibitem{ORR12}
K.~Onak, D.~Ron, M.~Rosen, R.~Rubinfeld, A near-optimal sublinear-time
  algorithm for approximating the minimum vertex cover size, in: Proceedings of
  the twenty-third annual ACM-SIAM symposium on Discrete Algorithms, SIAM,
  2012, pp. 1123--1131.

\bibitem{ba2009sublinear}
K.~D. Ba, H.~L. Nguyen, H.~N. Nguyen, R.~Rubinfeld, Sublinear time algorithms
  for earth mover's distance, Theory of Computing Systems 48~(2) (2011)
  428--442.
\newblock \href {https://doi.org/10.1007/s00224-010-9265-8}
  {\path{doi:10.1007/s00224-010-9265-8}}.

\bibitem{MS13}
A.~McGregor, D.~Stubbs, Sketching earth-mover distance on graph metrics, in:
  P.~Raghavendra, S.~Raskhodnikova, K.~Jansen, J.~D.~P. Rolim (Eds.),
  Approximation, Randomization, and Combinatorial Optimization. Algorithms and
  Techniques, Lecture Notes in Computer Science, Vol. 8096, Springer, Berlin,
  Heidelberg, 2013, pp. 274--286.
\newblock \href {https://doi.org/10.1007/978-3-642-40328-6\_20}
  {\path{doi:10.1007/978-3-642-40328-6\_20}}.

\bibitem{Yao77}
A.~C.-C. Yao, Probabilistic computations: Toward a unified measure of
  complexity, in: 18th Annual Symposium on Foundations of Computer Science,
  1977, pp. 222--227.
\newblock \href {https://doi.org/10.1109/SFCS.1977.24}
  {\path{doi:10.1109/SFCS.1977.24}}.

\bibitem{hoeffding63}
W.~Hoeffding, \href{http://www.jstor.org/stable/2282952}{Probability
  inequalities for sums of bounded random variables}, Journal of the American
  Statistical Association 58~(301) (1963) 13--30.
\newline\urlprefix\url{http://www.jstor.org/stable/2282952}

\bibitem{OllV12}
Y.~Ollivier, C.~Villani, \href{https://doi.org/10.1137/11085966X}{A curved
  brunn--minkowski inequality on the discrete hypercube, or: What is the ricci
  curvature of the discrete hypercube?}, SIAM Journal on Discrete Mathematics
  26~(3) (2012) 983--996.
\newblock \href {http://arxiv.org/abs/https://doi.org/10.1137/11085966X}
  {\path{arXiv:https://doi.org/10.1137/11085966X}}, \href
  {https://doi.org/10.1137/11085966X} {\path{doi:10.1137/11085966X}}.
\newline\urlprefix\url{https://doi.org/10.1137/11085966X}

\bibitem{Ga02}
R.~J. Gardner, The {B}runn-{M}inkowski inequality, Bulletin of American
  Mathematical Society 39~(3) (2002) 355--405.
\newblock \href {https://doi.org/10.1090/S0273-0979-02-00941-2}
  {\path{doi:10.1090/S0273-0979-02-00941-2}}.

\bibitem{cms01}
D.~Cordero-Erausquin, R.~J. McCann, M.~Schmuckenschl\"ager, A riemannian
  interpolation inequality {\`{a}} la borell, brascamp and lieb, Inventiones
  Mathematicae 146 (2001) 219--257.
\newblock \href {https://doi.org/10.1007/s002220100160}
  {\path{doi:10.1007/s002220100160}}.

\bibitem{cms06}
D.~Cordero-Erausquin, R.~J. McCann, M.~Schmuckenschl\"ager,
  \href{https://afst.centre-mersenne.org/articles/10.5802/afst.1132/}{Pr\'ekopa{\textendash}leindler
  type inequalities on {Riemannian} manifolds, {Jacobi} fields, and optimal
  transport}, Annales de la Facult\'e des sciences de Toulouse :
  Math\'ematiques Ser. 6, 15~(4) (2006) 613--635.
\newblock \href {https://doi.org/10.5802/afst.1132}
  {\path{doi:10.5802/afst.1132}}.
\newline\urlprefix\url{https://afst.centre-mersenne.org/articles/10.5802/afst.1132/}

\bibitem{PS82}
C.~H. Papadimitriou, K.~Steiglitz, Combinatorial optimization: algorithms and
  complexity, Prentice-Hall, Inc., NJ, USA, 1982.

\end{thebibliography}

\end{document}